\documentclass[onefignum,onetabnum]{siamart171218}

\usepackage{lipsum}
\usepackage{amsfonts}
\usepackage{graphicx}
\usepackage{epstopdf}
\usepackage{algorithmic}
\ifpdf
  \DeclareGraphicsExtensions{.eps,.pdf,.png,.jpg}
\else
  \DeclareGraphicsExtensions{.eps}
\fi

\newsiamremark{remark}{Remark}
\newsiamremark{assumption}{Assumption}
\newsiamremark{hypothesis}{Hypothesis}
\crefname{hypothesis}{Hypothesis}{Hypotheses}
\newsiamthm{claim}{Claim}

\headers{Epidemic Spread under Game-Theoretic Activation}{A. R. Hota, T. Sneh, and K. Gupta}

\title{Impacts of Game-Theoretic Activation on Epidemic Spread over Dynamical Networks}

\author{Ashish R. Hota, Tanya Sneh, \and Kavish Gupta\thanks{The authors are with the Department of Electrical Engineering, IIT Kharagpur. Email: \email{ahota@ee.iitkgp.ac.in,tanyasneh@gmail.com,kavishgupta1999@gmail.com}.}}
\usepackage{amsopn}

\ifpdf
\hypersetup{
  pdftitle={Game-Theoretic Activation and Epidemic Spread},
  pdfauthor={A. R. Hota, T. Sneh, and K. Gupta}
}
\fi

\usepackage{graphics} 
\usepackage{tikz}
\usetikzlibrary{positioning}
\usetikzlibrary{arrows}
\usetikzlibrary{automata}
\usepackage{cite}

\usepackage{tcolorbox}
\usepackage{epstopdf}
\usepackage{epsfig} 
\usepackage{subcaption}
\usepackage{psfrag}
\usepackage{caption}
\usepackage{lipsum}
\usepackage{amsmath} 
\usepackage{amssymb}  
\usepackage{mathbbol}
\usepackage{tcolorbox}

\usepackage{wrapfig}
\usepackage{mathrsfs}
\usepackage{url}
\usepackage{amsmath,amsfonts,amssymb,color}

\usetikzlibrary{shapes,arrows,calc}
\usepackage{marvosym}

\newcommand{\ignore}[1]{}

\DeclareMathOperator*{\avg}{avg}

\newcommand{\oprocendsymbol}{\hbox{$\square$}}
\newcommand{\oprocend}

\DeclareMathAlphabet{\mathcal}{OMS}{cmsy}{m}{n}      

\newcommand{\FF}{\mathcal{F}}

\newcommand{\Eb}{\mathbb{E}}
\newcommand{\Pb}{\mathbb{P}}
\newcommand{\Rb}{\mathbb{R}}

\newcommand{\Ib}{\mathbb{I}}

\newcommand{\VV}{\mathcal{V}}
\newcommand{\DD}{\mathcal{D}}

\newcommand{\St}{\mathtt{S}}
\newcommand{\SXt}{\mathtt{SX}}
\newcommand{\SAXt}{\mathtt{SAX}}
\newcommand{\SYt}{\mathtt{SY}}
\newcommand{\Xt}{\mathtt{X}}
\newcommand{\Yt}{\mathtt{Y}}
\newcommand{\Rt}{\mathtt{R}}
\newcommand{\At}{\mathtt{A}}
\newcommand{\Nt}{\mathtt{N}}
\newcommand{\NEt}{\mathtt{NE}}
\newcommand{\Zt}{\mathtt{Z}}

\newcommand{\myclearpage}{\clearpage}
\allowdisplaybreaks

\newcommand{\rev}[1]{{#1}}

\begin{document}

\maketitle

\begin{abstract}
We investigate the evolution of epidemics over dynamical networks when nodes choose to interact with others in a selfish and decentralized manner. Specifically, we analyze the susceptible-asymptomatic-infected-recovered (SAIR) epidemic in the framework of activity-driven networks with heterogeneous node degrees and time-varying activation rates, and derive both individual and degree-based mean-field approximations of the exact state evolution. We then present a game-theoretic model where nodes choose their activation probabilities in a strategic manner using current state information as feedback, and characterize the {\it quantal response equilibrium} (QRE) of the proposed setting. We then consider the activity-driven susceptible-infected-susceptible (SIS) epidemic model, characterize equilibrium activation probabilities and analyze epidemic evolution in closed-loop. Our numerical results provide compelling insights into epidemic evolution under game-theoretic activation. Specifically, for the SAIR epidemic, we show that under suitable conditions, the epidemic can persist, as any decrease in infected proportion is counteracted by an increase in activity rates by the nodes. For the SIS epidemic, we show that in regimes where there is an endemic state, the infected proportion could be significantly smaller under game-theoretic activation if the loss upon infection is sufficiently high.
\end{abstract}


% REQUIRED

\section{Introduction}
Infectious diseases exploit the interaction among human beings and spread through society, sometimes infecting millions of people across the globe within a few months. In the absence of effective medicines and vaccines, reducing the social interaction between individuals is critical to contain the spread of highly infectious diseases. As observed in the context of COVID-19, Governments across the world have imposed various restrictions on human interaction and travel in the hopes of containing the spread of the epidemic and to save lives, often incurring significant economic distress \cite{coibion2020cost}. However, the success of these measures depends critically on whether individuals or groups comply with them. It is often infeasible for Governments to perfectly enforce these restrictions, and as observed in case of COVID-19, there have been instances where the infections have grown during periods of lockdown while in different regions of the world, infections have been contained in the absence of any significant restrictions on human mobility \cite{pachetti2020impact,dahlberg2020effects}. 

Therefore, it is critical to understand how humans make decisions regarding their level of interaction during the spread of an epidemic and the impacts of such decisions on the evolution of the epidemic. In this work, we rigorously investigate this problem by modeling (i) the evolution of the epidemic in the framework of activity-driven networks and (ii) the human decision-making process in a game-theoretic framework that allows the decision-makers to have bounded rationality.

\subsection{Related work}

The susceptible-infected-recovered (SIR) and susceptible-infected-susceptible (SIS) epidemic models remain as two fundamental mathematical models of epidemic dynamics \cite{hethcote2000mathematics,pastor2015epidemic,nowzari2016analysis}. In both epidemic models, susceptible nodes become infected in a probabilistic manner if they come in contact with infected nodes. While in the SIR epidemic, an infected node develops immunity from the disease, in the SIS epidemic, there is a possibility of re-infection. Much of the prior work has investigated the above two epidemic dynamics on well-mixed populations and on static networks \cite{pastor2015epidemic}. However, during the prevalence of an epidemic, people modify their interaction patterns in order to protect themselves from becoming infected. In other words, the network evolves in response to the epidemic and in a time-scale that is comparable to the epidemic evolution. Accordingly, several recent works have focused on understanding epidemic dynamics on temporal or dynamical networks \cite{masuda2017temporal,pare2017epidemic,enright2018epidemics}.

In this work, we consider the framework of activity-driven networks (ADNs) \cite{perra2012activity,zino2017analytical,masuda2017temporal}. In the framework of ADNs, nodes activate with a certain probability (possibly depending on their epidemic state) and form connections with other nodes at random. Once the connections are formed, the epidemic states of the nodes \rev{change} according to the SIR or SIS model. Eventually, the connections are discarded. Several recent papers have investigated SIR and SIS epidemics on ADNs. Specifically, \cite{zino2020analysis,nadini2018epidemic,zino2020assessing} consider continuous-time evolution of epidemic states and network topology and derive conditions under which the epidemic is eradicated. In contrast, \cite{ogura2019optimal} derives mean-field approximations of discrete-time evolution of the SIS epidemic on ADNs. The above settings assume that (i) all nodes are (degree) homogeneous, (ii) nodes activate at a reduced rate without any strategic (game-theoretic) decision-making, (iii) optimal activity rates are computed via centralized optimization, and (iv) the impacts of asymptomatic carriers (as is the case with COVID-19) are not considered. In our prior work \cite{hota2020generalized}, we generalized the setting in \cite{ogura2019optimal} to include asymptomatic carriers and degree heterogeneity in the SIS epidemic model. 

In contrast with the above works where the activity rates are specified in an exogenous manner or optimized by a central authority, humans take these decisions in a strategic manner by evaluating the (social and economic) benefits of such interactions with the possible risk of becoming infected. Therefore, in this paper, we propose a game-theoretic framework where nodes choose their activation probabilities in order to maximize a suitably defined state-dependent utility function. Earlier work on game-theoretic reduction in contact rates has primarily focused on the setting with a well-mixed population \cite{bauch2004vaccination,theodorakopoulos2013selfish,reluga2010game,fenichel2011adaptive,choi2020optimal,chang2020game}. Recent papers have analyzed the impacts of game-theoretic protection against networked epidemic models \cite{omic2009protecting,trajanovski2017designing,hota2019game,hota2019game2}; these works focus on the SIS epidemic model (on static networks) and define the cost in terms of the infection probability at the endemic state of the epidemic. In a related work \cite{huang2019differential}, the authors consider a differential game framework for the SIS epidemic. A few recent works consider game-theoretic models of social distancing in a comparable time-scale as epidemic evolution; specifically \cite{cho2020mean,dasaratha2020virus} consider adaptive activation rates in well-mixed populations while \cite{lagos2020games} studies the impacts of local and aggregate information on the equilibrium strategies in a networked SIR epidemic model. \rev{In a closely related work \cite{eksin2016disease}, the authors considered the SIS epidemic model on static networks where nodes take decisions on whether to interact or not as a function of the current epidemic prevalence and analyzed the pure Nash equilibria.} To the best of our knowledge, ours is one of the first papers to rigorously analyze game-theoretic or strategic choice of activation in temporal networks, particularly in the ADN framework that \rev{incorporates} (i) bounded rationality in decision-making, (ii) heterogeneous node degrees, and (iii) asymptomatic carriers. We now summarize the contributions of this work. 

\subsection{Summary of contributions} We first define an activity-driven and adaptive analogue of the susceptible-asymptomatic-infected-recovered (SAIR) epidemic model \cite{robinson2013model, ansumali2020modelling,stella2020role} (Section \ref{section:a-siyr}). In the SAIR epidemic, nodes in the asymptomatic state can cause new infections and can potentially recover without ever exhibiting symptoms. Consequently, this is an appropriate framework to model diseases such as COVID-19 (compared to the classical SIR epidemic model) \cite{ansumali2020modelling}. Furthermore, both SIR and SEIR (with state E pertaining to the exposed state) models are special cases of the SAIR epidemic model as we show later. We consider a discrete-time model of epidemic and network evolution analogous to the setting in \cite{ogura2019optimal}, and refer our model as the A-SAIR epidemic. We derive individual-based and degree-based mean-field approximations of the exact Markovian evolution of the epidemic under arbitrary state and time-dependent activation rates and heterogeneous node degrees. 

In Section \ref{section:activation_game}, we formally introduce the {\it activation game} where each node decides whether to activate or not at each time. Nodes who activate receive a benefit, while susceptible nodes experience a risk of becoming infected due to their interactions. We assume that nodes in the asymptomatic state are not aware of being infected and continue to behave as susceptible nodes. Symptomatic nodes no longer face the risk of becoming infected, but incur a higher cost if they activate. We assume that nodes choose their actions according to the logit choice model which captures bounded rationality prevalent in human decision-making \cite{luce2012individual,greene2009discrete}, and consider the notion of quantal response equilibrium (QRE) \cite{mckelvey1995quantal} as the solution concept.\footnote{The QRE is a generalization of the Nash equilibrium which allows for stochastic uncertainty in the utility functions of the nodes, has strong behavioral foundations and has been shown to explain/predict observed human behavior in many applications \cite{mckelvey1995quantal,goeree2002quantal,haile2008empirical}.} We define the infection risk of the nodes via (degree-based) mean-field approximations and characterize the activation probabilities of the nodes at the QRE in the A-SAIR epidemic model. Our analysis shows that nodes use information regarding the current epidemic prevalence as feedback to determine their activity rates at the QRE. 

In Section \ref{section:a-sis}, we consider the activity-driven SIS epidemic model, derive the individual and degree-based mean-field approximations, and characterize the activation probabilities of the nodes at the QRE; the latter being a nonlinear function of the current epidemic prevalence. We then analyze the nonlinear closed-loop dynamics under game-theoretic activation probabilities in several special cases of interest. 

Finally, in Section \ref{section:simulation}, we report extensive simulations of the evolution of the epidemic states under game-theoretic activation decisions. Our results highlight the accuracy of the mean-field approximations, the impacts of asymptomatic carriers and heterogeneity in degree distributions on the evolution of the A-SAIR epidemic. Specifically, we show that under certain conditions, selfish and decentralized activation decisions can lead to long-term persistence of the disease in the population because any decline in the infected population is counteracted by an increase in activation probabilities, and vice versa. For the A-SIS epidemic, we show that if the loss due to infection is sufficiently high, the proportion of infected nodes at the endemic state can be made arbitrarily small under game-theoretic activation. We conclude with a discussion on directions for future research.
\section{Activity-Driven Adaptive SAIR Epidemic Model}
\label{section:a-siyr}

%%%% Figure %%%%%

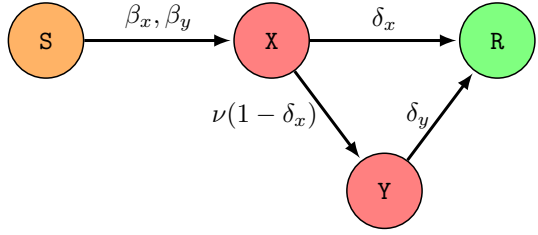
\begin{figure}[tb]
\centering
\begin{tikzpicture}[font=\sffamily]

% Setup the style for the states
\tikzset{node style/.style={state, minimum width=1cm, line width=0.2mm, fill=orange!60!white}}
\tikzset{node style1/.style={state, minimum width=1cm, line width=0.2mm, fill=red!50!white}}
\tikzset{node style2/.style={state, minimum width=1cm, line width=0.2mm, fill=green!50!white}}
        % Draw the states
\node[node style] at (0, 0)     (St)     {$\St$};
\node[node style1] at (3, 0)     (At)     {$\Xt$};
\node[node style2] at (6, 0) (Rt)  {$\Rt$};
\node[node style1] at (4.5, -2) (Yt) {$\Yt$};
        % Connect the states with arrows
        \draw[every loop,
              auto=right,
              line width=0.4mm,
              >=latex,
              draw=black,
              fill=black]
            %(St)     edge[bend right=20]            node {0.1} (Yt)
            (St)     edge[bend right=0, auto=left] node[above] {$\beta_x, \beta_y$} (At)
            (At)     edge[bend right=0] node[above] {$\delta_x$} (Rt)
            (At) edge node[left] {\rev{$\nu (1-\delta_x)$}} (Yt)
            %(Yt) edge[bend right=20]            node {0.2} (At)
            (Yt) edge node[left] {$\delta_y$} (Rt);
    \end{tikzpicture}
\caption{\footnotesize Probabilistic evolution of states in the A-SAIR epidemic model. Self-loops are omitted for better clarity. See Definition \ref{def:siyr} for the formal definition. Red indicates that both $\Xt$ and $\Yt$ are infected (as well as infectious) states.}
\label{fig:siys_tran}
\end{figure}

In this section, we formally define the activity-driven adaptive SAIR (A-SAIR) epidemic model and approximate the probabilistic evolution of the epidemic states via individual-based and degree-based mean-field (DBMF) approximations. Let $\VV = \{v_1, v_2, \ldots, v_n\}$ denote the set of $n$ individuals or nodes. Each node remains in one of the four possible states: susceptible ($\St$), asymptomatic ($\Xt$), infected with symptoms or symptomatic ($\Yt$) and recovered ({$\Rt$}). We assume that the recovered state includes the nodes that are deceased. Both asymptomatic and symptomatic individuals are infectious, which captures the characteristics of epidemics such as COVID-19. The states evolve in discrete-time. If at time $k \in \{0,1,\ldots\}$, node $v_i$ is susceptible (respectively, asymptomatic, symptomatic and recovered), we denote this by $v_i(k) \in \St$ (respectively, $v_i(k) \in \Xt$, $v_i(k) \in \Yt$ and $v_i(k) \in \Rt$). Given a network or contact pattern, the probabilistic state evolution is defined below.

\begin{definition}\label{def:siyr}
Let $\beta_x, \beta_y, \delta_x, \delta_y, \nu \in [0, 1]$ be constants pertaining to infection, recovery and transition rates. The state of each node $v_i$ evolves as follows.
\begin{enumerate}
\item \rev{If $v_i(k) \in \St$, then $v_i(k+1) \in \Xt$ if it is infected by at least one neighbor. The probability of infection is $\beta_x$ for an asymptomatic neighbor and $\beta_y$ for a symptomatic neighbor. The infection caused by each neighbor is independent of the infection caused by others.}
\item If $v_i(k) \in \Xt$, then $v_i(k+1) \in \Rt$ with probability $\delta_x$ and $v_i(k+1) \in \Yt$ with probability $\nu(1-\delta_x)$. 
\item If $v_i(k) \in \Yt$, then $v_i(k+1) \in \Rt$ with probability $\delta_y$. 
\end{enumerate}
The state remains unchanged otherwise. \hfill \oprocendsymbol
\end{definition}

The possible transitions of the states are illustrated in Figure \ref{fig:siys_tran}. Thus, in our model, both asymptomatic and symptomatic nodes can potentially infect a susceptible node, albeit with possibly different probabilities ($\beta_x$ and $\beta_y$, respectively). Upon being infected, a susceptible node becomes asymptomatic. From there on, it can either get cured and become recovered with probability $\delta_x$, and if not, it transitions to the symptomatic state with probability $\nu$. Thus $\nu^{-1}$ captures the delay in onset of symptoms. The recovery/curing rate for symptomatic nodes is $\delta_y$. Nodes in state $\Xt$ act as asymptomatic carriers of the disease. 

We denote the degree of node $v_i$ as $d_i \in \DD$ where $\DD \subset \{1,2,\ldots,d_{\max}\}$ denotes the set of all degrees of the nodes. The maximum degree $d_{\max}$ is assumed to be finite. The proportion of nodes with degree $d \in \DD$ is denoted by $m_d$ with $\sum_{d \in \DD} m_d = 1$. We denote the average or mean degree by $d_{\avg} := \sum_{d \in \DD} dm_d$. We now formally define the state-dependent evolution of the network or contact pattern for our setting. 

\begin{definition}\label{def:act_asiyr}
For each node $v_i \in \VV$, let $\pi_{\Zt,i}(k) \in [0,1]$ denote the activation probability of node $v_i$ at time $k$ in state $\Zt \in \{\St,\Xt,\Yt,\Rt\}$. Let $\beta_x, \beta_y, \delta_x, \delta_y, \nu \in [0, 1]$ be constants pertaining to infection, recovery and transition rates as before. The A-SAIR model is defined by the following procedures:
\begin{enumerate}
\item At the initial time $k = 0$, each node is in one of the four possible states.
\item At each time $k = 1, 2, \ldots$, each node $v_i$ activates with probability $\pi_{\Zt,i}(k)$ where $\Zt$ is the state of $v_i$ at $k$ independently of other nodes .
\item Node $v_i$, upon activation, chooses $d_i$ other nodes randomly and uniformly among all nodes and forms edges with them. If $v_j$ is chosen by $v_i$, an (undirected) edge $(v_i,v_j)$ is created with probability $1$.
\item Once the edges are formed, the states of the nodes get updated following Definition \ref{def:siyr}. 
\item These edges are discarded at time $k + 1$. Steps 2-4 are repeated for $k \geq 1$. \hfill \oprocendsymbol
\end{enumerate} 
\end{definition}

\sloppy
Note that the activation probabilities are potentially time-varying, state-dependent and heterogeneous across nodes. 

\begin{remark}
\rev{In our model, if a node $v_i$ does not activate, it may still interact with others; specifically when another node $v_j$ activates and chooses $v_i$ to form an edge. This feature captures settings where an individual may not go out to interact with others, but can not prevent others from paying a visit. In contrast, when node $v_i$ activates, it forms $d_i$ edges on its own, in addition to the edges that may be formed by other activating nodes choosing $v_i$.} \hfill \oprocendsymbol
\end{remark}

The following remark highlights the generality of the proposed model.

\begin{remark}\label{remark:asiyr_special}
The above model captures several epidemic models as special cases.
\begin{enumerate}
\item When $\beta_x = \delta_x = 0$, then nodes in asymptomatic state $\Xt$ necessarily become symptomatic before recovering. Furthermore, nodes in $\Xt$ do not cause any new infections. Thus, we recover the activity-driven analogue of the classical SEIR epidemic model with nodes in $\Xt$ being analogous to the ``exposed" state. 
\item When $\nu = 0$ and none of the nodes are in state $\Yt$ at time $k=0$, we recover the activity-driven analogue of the classical SIR epidemic with $\Xt$ being the state of infected nodes. \hfill \oprocendsymbol
\end{enumerate}
\end{remark}

\subsection{State Evolution and Mean-Field Approximations}\label{section:asiyr-mfa}

We now analyze the state evolution and derive mean-field approximations of the above general model. We define (indicator) random variables $S_i(k), X_i(k), Y_i(k)$ and $R_i(k)$ associated with node $v_i$ that take values in the set $\{0,1\}$ with $S_i(k) = 1$ if $v_i(k) \in \St$, $X_i(k) = 1$ if $v_i(k) \in \Xt$, $Y_i(k) = 1$ if $v_i(k) \in \Yt$ and $R_i(k) = 1$ if $v_i(k) \in \Rt$. Since a node can only be in one of the four possible states, we have $S_i(k) + X_i(k) + Y_i(k) + R_i(k) = 1$. Similarly, we define a $\{0,1\}$-valued random variable $A_{ij}(k)$ which takes value $1$ if the edge $(v_i,v_j)$ exists at time $k$. We also denote by $N_x$ a Bernoulli random variable that takes value $1$ with probability $x \in [0,1]$. The exact state evolution of a node $v_i$ under the A-SAIR epidemic model can now be formally stated as:
\begin{subequations}\label{eq:asiyr_markov}
\begin{align}
& S_i(k+1) = S_i(k) \underset{j \neq i}{\Pi} \left[ 1 - A_{ij}(k) (X_j(k)N_{\beta_x} + Y_j(k)N_{\beta_y}) \right], \label{eq:asiyr_markov_s}
\\ & X_i(k+1) = (1-N_{\delta_x})(1-N_{\nu}) X_i(k) \nonumber
\\ & \qquad \qquad \qquad + S_i(k) \big[1 - \underset{j \neq i}{\Pi} \left[ 1 - A_{ij}(k) (X_j(k)N_{\beta_x} + Y_j(k)N_{\beta_y}) \right]\big],  \label{eq:asiyr_markov_x}
\\ & Y_i(k+1) = (1-N_{\delta_x})N_{\nu} X_i(k) + (1-N_{\delta_y})Y_i(k), \label{eq:asiyr_markov_y}
\\ & R_i(k+1) = R_i(k) + N_{\delta_x}X_i(k) + N_{\delta_y}Y_i(k). \label{eq:asiyr_markov_r}
\end{align}
\end{subequations}
It is easy to see that $S_i(k+1) + X_i(k+1) + Y_i(k+1) + R_i(k+1) = S_i(k) + X_i(k) + Y_i(k) + R_i(k) = 1$. We denote the probability of node $v_i$ being susceptible at time $k$ by $s_i(k)$, i.e., $s_i(k) := \Pb(v_i(k) \in \St) = \Pb(S_i(k) = 1) = \Eb[S_i(k)]$. Similarly, we define $x_i(k) := \Eb[X_i(k)]$ , $y_i(k) := \Eb[Y_i(k)]$ and $r_i(k) := \Eb[R_i(k)]$.

Note that the infection states follow a Markov process with a $4^n \times 4^n$ transition probability matrix. While analyzing the exact state evolution of this model is computationally intractable, we first present a mean-field approximation of the evolution of the epidemic states for each individual node. 

\begin{theorem}[Individual-based mean-field approximation of the A-SAIR epidemic]\label{thm:asiyr_main}
Consider the A-SAIR epidemic model defined in Definition \ref{def:act_asiyr}.  Let $\bar{d}_i = d_i/(n-1)$, and for all $i, j$, define the constants
\begin{align}
\beta^{ij}_x(k) & := \beta_x [1-(1-\pi_{\St,i}(k) \bar{d}_i)(1 - \pi_{\Xt,j}(k) \bar{d}_j)] \label{eq:betaxij_tv},
\\ \beta^{ij}_y(k) & := \beta_y [1-(1-\pi_{\St,i}(k) \bar{d}_i)(1 - \pi_{\Yt,j}(k) \bar{d}_j)]. \label{eq:betayij_tv}
\end{align}
Then,
\begin{subequations}\label{eq:mfa_linear_dynamics_asiyr}
\begin{align}
& s_i(k+1) \simeq s_i(k) \Big(1-\sum_{j \neq i} \big[\beta^{ij}_x(k) x_j(k) + \beta^{ij}_y(k) y_j(k) \big]\Big), 
\\ & x_i(k+1) \simeq (1-\delta_x) (1-\nu) x_i(k) + s_i(k) \sum_{j \neq i} \left[\beta^{ij}_x(k) x_j(k) + \beta^{ij}_y(k) y_j(k) \right], 
\\ & y_i(k+1) = (1-\delta_x) \nu x_i(k) + (1-\delta_y)y_i(k),
\\ & r_i(k+1) = r_i(k) + \delta_x x_i(k) + \delta_y y_i(k), 
\end{align}
\end{subequations}
for all nodes $v_i$ and $k \geq 0$. \hfill \oprocendsymbol
\end{theorem}

\begin{proof}
We compute expectation on both sides of \eqref{eq:asiyr_markov_x} and \eqref{eq:asiyr_markov_y} and obtain
\begin{subequations}\label{eq:asiyr_mfa}
\begin{align}
& x_i(k+1) = (1-\delta_x)(1-\nu) x_i(k) \nonumber
\\ & \qquad \qquad \qquad + \Eb\Big[S_i(k) \big[1 - \underset{j \neq i}{\Pi} \left[1 - A_{ij}(k) (X_j(k)N_{\beta_x} + Y_j(k)N_{\beta_y}) \right]\big]\Big], \label{eq:asiyr_mfa2_x}
\\ & y_i(k+1) = (1-\delta_x) \nu x_i(k) + (1-\delta_y)y_i(k),
\end{align}
\end{subequations}
For the product term in the R.H.S. of \eqref{eq:asiyr_mfa2_x}, the Weierstrass product inequality yields
\begin{align*}
& 1 - \underset{j \neq i}{\Pi} \left[ 1 - A_{ij}(k) (X_j(k)N_{\beta_x} + Y_j(k)N_{\beta_y}) \right] \leq \sum_{j\neq i} A_{ij}(k) (X_j(k)N_{\beta_x} + Y_j(k)N_{\beta_y}). 
\end{align*}
Consequently, we have
\begin{align}
& \Eb\Big[S_i(k) \big[1 - \underset{j \neq i}{\Pi} \left[ 1 - A_{ij}(k) (X_j(k)N_{\beta_x} + Y_j(k)N_{\beta_y}) \right]\big]\Big] \nonumber
\\ & \qquad \qquad \leq \!\sum_{j \neq i} \!\beta_x \Eb[A_{ij}(k) S_i(k) X_j(k)] \!+\! \beta_y \Eb[A_{ij}(k) S_i(k) Y_j(k)]. \label{eq:ineq_main}
\end{align}

We now focus on evaluating the expectation terms in the above equation. Recall that $A_{ij}(k)$ is a random variable that indicates the presence of the edge $(v_i,v_j)$ at time $k$ and is governed by the states and activation probabilities of nodes $v_i$ and $v_j$ according to Definition \ref{def:act_asiyr}. In order to compute the expectation terms, we introduce the following notation for events of interest:
\begin{align}
& \SXt^k_{ij} = ``v_i(k) \in \St \text{ and } v_j(k) \in \Xt," \label{eq:def_SXtij}
\\ & \SYt^k_{ij} = ``v_i(k) \in \St \text{ and } v_j(k) \in \Yt," \label{eq:def_SYtij}
\\ & \Gamma^k_{i\to j} = \text{``$v_i$ is activated, and chooses } v_j \text{ as neighbor at time $k$."} \label{eq:def_Gammatij}
\end{align}
With the above notation in place, we have
\begin{subequations}\label{eq:sub_cond_exp}
\begin{align}
\Eb[A_{ij}(k) S_i(k) X_j(k)] & = \Pb(A_{ij}(k)=1, S_i(k) = 1, X_j(k) = 1) \nonumber
\\ & = \Pb(A_{ij}(k)=1|\SXt^k_{ij})\Pb(\SXt^k_{ij}), \label{eq:sub_cond_exp1}
\\ \Eb[A_{ij}(k) S_i(k) Y_j(k)] & = \Pb(A_{ij}(k)=1|\SYt^k_{ij})\Pb(\SYt^k_{ij}). \label{eq:sub_cond_exp2}
\end{align}
\end{subequations}
We now focus on the first equation above and note that
\begin{align*}
\Pb(A_{ij}(k)=1|\SXt^k_{ij}) & = \Pb(\Gamma^k_{i \to j} \cup \Gamma^k_{j \to i}| \SXt^k_{ij})
\\ & = \Pb(\Gamma^k_{i \to j} | \SXt^k_{ij}) + \Pb(\Gamma^k_{j \to i} | \SXt^k_{ij}) - \Pb(\Gamma^k_{i \to j} | \SXt^k_{ij}) \Pb(\Gamma^k_{j \to i} | \SXt^k_{ij})
\\ & = 1-[1-\Pb(\Gamma^k_{i \to j} | \SXt^k_{ij})][1-\Pb(\Gamma^k_{j \to i} | \SXt^k_{ij})]
\\ & = 1-[1-\pi_{\St,i}(k)\bar{d}_i][1-\pi_{\Xt,j}(k)\bar{d}_j]
\end{align*}
following Definition \ref{def:act_asiyr}. Note that the probability of an activated node $v_i$ choosing a specific node $v_j$ is $\bar{d}_i = \frac{d_i}{n-1}$. Similarly for events conditioned on $\SYt^k_{ij}$, we have 
\begin{align*}
\Pb(A_{ij}(k)=1|\SYt^k_{ij}) & = 1-[1-\Pb(\Gamma^k_{i \to j} | \SYt^k_{ij})] [1-\Pb(\Gamma^k_{j \to i} | \SYt^k_{ij})]
\\ & = 1-[1-\pi_{\St,i}(k)\bar{d}_i][1-\pi_{\Yt,j}(k)\bar{d}_j].
\end{align*}
Finally, we approximate $\Pb(\SXt^k_{ij})$ and $\Pb(\SYt^k_{ij})$ by assuming that the events $v_i(k) \in \St$ and $v_j(k) \in \Xt$ (or $v_j(k) \in \Yt$) are independent. Substituting the above in \eqref{eq:sub_cond_exp} yields
\begin{align*}
\Eb[A_{ij}(k) S_i(k) X_j(k)] & = [1-[1-\pi_{\St,i}(k)\bar{d}_i][1-\pi_{\Xt,j}(k)\bar{d}_j]] \Pb(\SXt^k_{ij}) 
\\ & \simeq [1-[1-\pi_{\St,i}(k)\bar{d}_i][1-\pi_{\Xt,j}(k)\bar{d}_j]] s_i(k) x_j(k), 
\\ \Eb[A_{ij}(k) S_i(k) Y_j(k)] & \simeq [1-[1-\pi_{\St,i}(k)\bar{d}_i][1-\pi_{\Yt,j}(k)\bar{d}_j]] s_i(k) y_j(k).
\end{align*}
The result now follows upon substituting the above in \eqref{eq:ineq_main} using the definition of $\beta^{ij}_x$ and $\beta^{ij}_y$, and taking expectations on both sides of \eqref{eq:asiyr_markov_s} and \eqref{eq:asiyr_markov_r}. 
\end{proof}

\begin{remark}\label{remark:WPI_doerr}
\rev{While evaluating the evolution of the probability of a node being in different epidemic states in the above proof, we have relied on two approximations. 
\begin{itemize}
\item We have assumed that the probability of two nodes being susceptible and infected respectively are independent, i.e., $\Pb(\SXt^k_{ij}) \simeq s_i(k)x_j(k)$ and $\Pb(\SYt^k_{ij}) \simeq s_i(k)y_j(k)$. This is reasonable when the number of nodes is large, and because the network is formed anew in each time step in a probabilistic manner.
\item We have applied the Weierstrass product inequality to upper bound the event that corresponds to a susceptible node becoming infected. This inequality \cite[Lemma 1.4.8]{doerr2020probabilistic} states that for $p_1, \ldots, p_N \in [0,1]$, we have 
$$ 1 - \sum^N_{i=1} p_i \leq \Pi^N_{i=1} (1-p_i) \leq 1-\! \sum^N_{i=1} p_i + \frac{1}{2} \big(\sum^N_{i=1} p_i\big)^2 \!= 1 - \big(\sum^N_{i=1} \! p_i\big) \big(1-\frac{1}{2} \sum^N_{i=1} p_i\big). $$
In the proof, we have used the left hand side inequality. Nevertheless, the inequality on the right hand side suggests that the (left hand side) inequality is tight when $\sum^N_{i=1} p_i \ll 2$. While the summand in the proof $\sum_{j\neq i} A_{ij}(k) (X_j(k)N_{\beta_x} + Y_j(k)N_{\beta_y})$ is a random variable, its expectation is precisely $\sum_{j \neq i} \left[\beta^{ij}_x(k) x_j(k) + \beta^{ij}_y(k) y_j(k) \right]$. Thus, the bound is tight (on average) when $\sum_{j \neq i} \left[\beta^{ij}_x(k) x_j(k) + \beta^{ij}_y(k) y_j(k) \right]$ is much smaller than $2$ and $n$ is large. A (weak) sufficient condition for this is $(\beta_x + \beta_y) d_{\max} \ll 1$.  \hfill \oprocendsymbol
\end{itemize}} 
\end{remark}

%%%%%%%%%%%%%%

The result shows that the evolution of the probability of a node being in any of the four epidemic states is approximated by the dynamics stated in \eqref{eq:mfa_linear_dynamics_asiyr} with state dimension $4n$ which is considerably smaller than the exact Markovian evolution with dimension $4^n$. We now leverage the above result to obtain a degree-based mean field (DBMF) approximation of the A-SAIR epidemic model when all nodes of a given degree choose an identical state-dependent activation probability and are equally likely to be in any of the epidemic states. The accuracy of the DBMF approximation improves when the number of nodes increases. 

\begin{definition}\label{thm:asiyr_dbmf}
Consider the A-SAIR epidemic model defined in Definition \ref{def:act_asiyr}. Let all nodes with degree $d \in \DD$ choose an identical (state-dependent) activation probability denoted by $\pi_{d,\Zt}(k)$ with $\Zt \in \{\St,\Xt,\Yt,\Rt\}$. Let $\bar{d} = d/(n-1)$, and for all $d, t \in \DD$, define the constants
\begin{align}
\beta^{dt}_x(k) & := \beta_x [1-(1-\pi_{d,\St}(k) \bar{d})(1 - \pi_{t,\Xt}(k) \bar{t})] \label{eq:betaxij_tv_dbmf},
\\ \beta^{dt}_y(k) & := \beta_y [1-(1-\pi_{d,\St}(k) \bar{d})(1 - \pi_{t,\Yt}(k) \bar{t})]. \label{eq:betayij_tv_dbmf}
\end{align}
Let $s_d(k), x_d(k), y_d(k), r_d(k)$ denote the proportion of degree $d$ nodes that are susceptible, asymptomatic, symptomatic and recovered at time $k$, respectively. Then, we define the DBMF approximation of the A-SAIR epidemic as the dynamics given by
\begin{subequations}\label{eq:dbmf_asiyr}
\begin{align}
& s_d(k+1) = s_d(k) \Big(1-\sum_{t \in \DD} nm_t \big[\beta^{dt}_x(k) x_t(k) + \beta^{dt}_y(k) y_t(k) \big]\Big), 
\\ & x_d(k+1) = (1-\delta_x)\! (1-\nu) x_d(k)\! + s_d(k) \! \sum_{t \in \DD} \! nm_t \big[\beta^{dt}_x(k) x_t(k) \!+ \beta^{dt}_y(k) y_t(k) \big], 
\\ & y_d(k+1) = (1-\delta_x) \nu x_d(k) + (1-\delta_y)y_d(k),
\\ & r_d(k+1) = r_d(k) + \delta_x x_d(k) + \delta_y y_d(k), 
\end{align}
\end{subequations}
for all degrees $d$ and $k \geq 0$. \hfill \oprocendsymbol
\end{definition}

The above approximation is inspired by the analysis in Theorem \ref{thm:asiyr_main}. Specifically, for a node $v_i$ with degree $d$, we have
\rev{\begin{align*}
& \sum_{j \neq i} \big[\beta^{ij}_x(k) x_j(k) + \beta^{ij}_y(k) y_j(k) \big] 
\\ = & \sum_{t \in \DD, t \neq d} nm_t \big[\beta^{dt}_x(k) x_t(k) + \beta^{dt}_y(k) y_t(k) \big] + (n-1)m_d \big[\beta^{dd}_x(k) x_d(k) + \beta^{dd}_y(k) y_d(k) \big]
\\ \simeq & \sum_{t \in \DD} nm_t \big[\beta^{dt}_x(k) x_t(k) + \beta^{dt}_y(k) y_t(k) \big],
\end{align*}}
as all nodes with degree $t$ contribute an identical quantity in the summation and the number of nodes with degree $t$ is $nm_t$. \rev{When the number of nodes $n$ is large, $(n-1)m_d \simeq nm_d$ which yields the above approximation.}

Note that when we all nodes with a given degree $d$ have an identical activation probability, we denote the activation probability as $\pi_{d,\Zt}$ while when we consider the activation probability of a node $v_i$, we denote it as $\pi_{\Zt,i}$. In the following section, we rely on the DBMF approximation to analyze strategic activation by the nodes. 

\begin{remark}
The dynamics under the mean-field approximations (\eqref{eq:mfa_linear_dynamics_asiyr} and \eqref{eq:dbmf_asiyr}) indicate that the proportion of susceptible nodes is monotonically decreasing and any equilibrium must be free from infected nodes. When the activation probabilities are constants, the dynamics (\eqref{eq:mfa_linear_dynamics_asiyr} and \eqref{eq:dbmf_asiyr}) are analogous to the dynamics of the classical SAIR epidemic analyzed in \cite{ansumali2020modelling,stella2020role}. We omit further discussions on this as it is not a major contribution of our work. \hfill \oprocendsymbol
\end{remark}
\section{To Activate or Not: Game-Theoretic Model}
\label{section:activation_game}

In this section, we formulate and analyze a game-theoretic setting where nodes choose their activation probabilities in a selfish and decentralized manner. We consider the A-SAIR epidemic model described above. At the beginning of time $k$, nodes decide whether to activate themselves or not in order to maximize their individual utilities at time $k$ in a myopic manner. The utility of a node is a function of its epidemic state, its chosen action and the states and actions of other nodes. Since the network is randomly created anew at each time step, all nodes with a given degree and epidemic state face an identical decision problem and as a result, we assume that they adopt the same strategy. We refer the collection of nodes with degree $d \in \DD$ as a subpopulation with degree $d$.

We now define the utility of a node in different epidemic states. We assume that if a node activates, it receives $1$ unit of benefit (due to social and economic interactions) irrespective of its epidemic state. In addition, we assume that a symptomatically infected node incurs a cost $c \in \Rb_{+}$ if it activates itself, which captures, for instance, deterrent placed by authorities for violating quarantine rules. If the node is susceptible, there is a risk that it will get infected when it comes in contact with another infected node and incurs a loss $L \in \Rb_{+}$ if it becomes infected. An asymptomatic node is unaware of its infection state and as a result, behaves as a susceptible node. Nodes in symptomatic and recovered states do not incur any risk of infection. 

Formally, a strategy profile for the nodes is denoted by $\pi := \{\pi_\St, \pi_\Xt, \pi_\Yt, \pi_\Rt\}$ with $\pi_{\Zt} := \{\pi_{d,\Zt}\}_{d \in \DD} \subseteq [0,1]^{|\DD|}$ where $\pi_{d,\Zt} \in [0,1]$ denotes the activation probability of a node with degree $d$ in epidemic state $\Zt \in \{\St,\Xt,\Yt,\Rt\}$. If a node decides to activate, we denote its action as $\At$; the decision to not activate is denoted as $\Nt$. At time $k$, we denote the global epidemic state as $e(k) := \{s_d(k),x_d(k),y_d(k),r_d(k)\}_{d \in \DD}$, i.e, the proportion of nodes in each epidemic state in each subpopulation. We assume that at a given time $k$, nodes are aware of the epidemic state either completely or at an aggregate level (to be made precise subsequently). The information available to the nodes is denoted as $\bar{e}(k)$. The utility of a node $v$ with degree $d \in \DD$ in different epidemic states at time $k$ are defined as:
\begin{subequations}\label{eq:utilitydef_asiyr}
\begin{align}
& u_{d,\St}(\At,\pi,\bar{e}(k)) = 1 - L\Rb_d(\pi, \bar{e}(k) | v(k) \in \St\At),
\\ & u_{d,\St}(\Nt,\pi, \bar{e}(k)) = - L\Rb_d(\pi, \bar{e}(k) | v(k) \in \St\Nt),
\\ & u_{d,\Xt}(\At,\pi, \bar{e}(k)) = u_{d,\St}(\At,\pi,\bar{e}(k)), \qquad u_{d,\Xt}(\Nt,\pi, \bar{e}(k)) = u_{d,\St}(\Nt,\pi, \bar{e}(k)),
\\ & u_{d,\Yt}(\At,\pi, \bar{e}(k)) = 1 - c, \qquad \qquad \qquad u_{d,\Yt}(\Nt,\pi, \bar{e}(k)) = 0,
\\ & u_{d,\Rt}(\At,\pi, \bar{e}(k)) = 1, \qquad \qquad \qquad \quad \quad u_{d,\Rt}(\Nt,\pi, \bar{e}(k)) = 0,
\end{align}
\end{subequations}
where $\Rb_d(\pi, \bar{e}(k) | v(k) \in \St\At)$ (respectively, $\Rb_d(\pi, \bar{e}(k) | v(k) \in \St\Nt)$) denotes the risk factor for a susceptible node of degree $d$ under strategy profile $\pi$ if it activates (respectively, does not activate) denoted by $\St\At$ (respectively, $\St\Nt$). Subsequently, we will make the notion of risk factor precise. 

We refer the above setting as the {\it activation game}. We assume that nodes choose their actions in a probabilistic manner following the logit choice model. \rev{Let the information available to the nodes be $\bar{e}(k)$ and the strategy profile be $\pi$. Then, a node with degree $d$ in epidemic state $\Zt$ chooses action $\At$ with probability}
\begin{equation}\label{eq:logit_def}
\sigma^\lambda_{d, \Zt,\At} (\pi, \bar{e}(k)) = \frac{e^{\lambda u_{d,\Zt}(\At,\pi, \bar{e}(k))}}{e^{\lambda u_{d,\Zt}(\At,\pi, \bar{e}(k))}+e^{\lambda u_{d,\Zt}(\Nt,\pi, \bar{e}(k))}}, \qquad \Zt \in \{\St,\Xt,\Yt,\Rt\}, d \in \DD,
\end{equation} 
where the parameter $\lambda$ captures the error in decision-making process of the nodes. As $\lambda \to 0$, both actions are chosen with probabilities approaching $0.5$ (nodes choose their actions completely randomly) while as $\lambda \to \infty$, the action with higher utility is chosen with probability approaching $1$ (nodes are perfectly rational). Intermediate values of $\lambda$ capture bounded rationality in the decisions made by the nodes \cite{mckelvey1995quantal}. 

\rev{Note that $\sigma^\lambda_{d, \Zt,\At} (\pi, \bar{e}(k))$ denotes the best response of a node with degree $d$ in epidemic state $\Zt$ to the strategy profile $\pi$. While the strategy profile $\pi$ specifies an activation probability $\pi_{d,\Zt}$ for such a node, it may not be the optimal activation probability for this node given its utility functions and the logit choice model. Consequently, such a strategy profile $\pi$ will not remain invariant. A strategy profile $\pi$ for which the specified activation probabilities ($\pi_{d,\Zt}$'s) coincide with the best responses ($\sigma^\lambda_{d,\Zt,\At}(\pi,\bar{e}(k))$'s) is referred to as a {\it quantal response equilibrium} of the game.}

\begin{definition}\label{def:qredef_asiyr}
A strategy profile $\pi^{\lambda}_\NEt(k) = \{\pi^\lambda_{\St,\NEt}(k), \pi^\lambda_{\Xt,\NEt}(k), \pi^\lambda_{\Yt,\NEt}(k), \pi^\lambda_{\Rt,\NEt}(k)\}$ is a quantal response equilibrium (QRE) with parameter $\lambda$ at time $k$ if
\begin{align*}
\pi^\lambda_{d, \Zt,\NEt}(k) & = \sigma^\lambda_{d,\Zt,\At} (\pi^\lambda_{\NEt}(k),\bar{e}(k)) 
\\ & = \frac{e^{\lambda u_{d,\Zt}(\At,\pi^{\lambda}_\NEt(k), \bar{e}(k))}}{e^{\lambda u_{d,\Zt}(\At,\pi^{\lambda}_\NEt(k), \bar{e}(k))}+e^{\lambda u_{d,\Zt}(\Nt,\pi^{\lambda}_\NEt(k), \bar{e}(k))}}
\\ & = \frac{e^{\lambda \Delta u_{d,\Zt}(\pi^{\lambda}_\NEt(k), \bar{e}(k))}}{e^{\lambda \Delta u_{d,\Zt}(\pi^{\lambda}_\NEt(k), \bar{e}(k))}+1},
\end{align*} 
where $\Delta u_{d,\Zt}(\pi, \bar{e}(k)) := u_{d,\Zt}(\At,\pi,\bar{e}(k)) - u_{d,\Zt}(\Nt,\pi,\bar{e}(k))$ for $\Zt \in \{\St,\Xt,\Yt,\Rt\}, d \in \DD$. \hfill \oprocendsymbol
\end{definition}

The above definition is essentially a consistency condition: \rev{the probability with which a node activates at the QRE coincides with its best response under the logit choice model \eqref{eq:logit_def}. In other words, the QRE strategy profile is a fixed point under the best response mapping.} Since there is a finite number of subpopulations and all nodes in a given subpopulation choose among finitely many actions, there always exists a QRE of the activation game \rev{as shown in \cite{mckelvey1995quantal} relying on fixed point arguments.}

\begin{remark}
The parameter $c$ can alternatively be viewed as an incentive given to infected nodes to self-isolate or self-quarantine themselves. In both interpretations (penalty and incentive), the difference in utility $\Delta u_{d,\Yt}(\pi, \bar{e}(k)) = 1-c$. \hfill \oprocendsymbol
\end{remark}

%%%%%%%%%%%%%%%%%%%%%%%%%%%%%%%%%%
%%%%%%%%%%%%%%%%%%%%%%%%%%%%%%%%%%

\subsection{Approximating the infection risk for susceptible nodes}
\label{section:riskfactor_mfa}

For a node $v_i$ with degree $d$, a natural choice for the risk factor $\Rb_d(\pi, \bar{e}(k)| v_i(k) \in \St\At)$ is the probability with which the susceptible node becomes infected upon activation, i.e., the quantity $\Pb(v_i(k+1) \in \Xt | v_i(k) \in \St\At; \pi)$. Computing this probability may be prohibitive in large networks as each node needs to know the epidemic states of every other node. Therefore, we compute an upper bound on this probability in a manner analogous to our derivation of mean-field approximations as
\begin{align}
& \Pb(v_i(k+1) \in \Xt | v_i(k) \in \St\At; \pi) \nonumber
\\ & \qquad = \Eb\big[1 - \underset{j \neq i}{\Pi} \left[ 1 - A_{ij}(k) (X_j(k)N_{\beta_x} + Y_j(k)N_{\beta_y}) \right]|v_i(k) \in \St\At; \pi\big] \nonumber
\\ & \qquad \leq \Eb\big[\sum_{j \neq i} \left[A_{ij}(k) (X_j(k)N_{\beta_x} + Y_j(k)N_{\beta_y}) \right]|v_i(k) \in \St\At; \pi\big] \nonumber
\\ & \qquad = \sum_{j \neq i} \beta_x \Eb[A_{ij}(k)X_j(k)|v_i(k) \in \St\At; \pi] + \beta_y \Eb[A_{ij}(k) Y_j(k)|v_i(k) \in \St\At; \pi], \label{eq:mfbound_siyr_1}
\end{align}
where the inequality is an application of Weierstrass product inequality. 

We now compute the above two expectation terms for a node $v_j$ with degree $t$. Recall from \eqref{eq:def_Gammatij} that $\Gamma^k_{i \to j}$ denotes the event that node $v_i$ activates and forms an edge with node $v_j$ at time $k$. We denote by $\SAXt^k_{ij}$ the event that $v_i$ is susceptible and activates itself while $v_j$ is asymptomatically infected at time $k$. We compute
\begin{align*}
\Eb[A_{ij}(k) X_j(k)|v_i(k) \in \St\At; \pi] & = \Pb(A_{ij}(k) = 1|v_i(k) \in \St\At, X_j(k) = 1; \pi) \Pb(X_j(k) = 1)
\\ & = \left[1\!-\![1\!-\Pb(\Gamma^k_{i \to j} | \SAXt^k_{ij}; \pi)][1\!-\Pb(\Gamma^k_{j \to i} | \SAXt^k_{ij}; \pi)]\right] x_t(k)
\\ & = \left[1-[1-\bar{d}][1-\bar{t}\pi_{t,\Xt}]\right] x_t(k)
\\ & = \left[\bar{t}\pi_{t,\Xt} + \bar{d}(1-\bar{t}\pi_{t,\Xt})\right]x_t(k),
\end{align*}
where $\bar{d} = d/(n-1)$ is the probability of node $v_i$ choosing a specific node $v_j$; the interpretation of $\bar{t}$ is analogous. Furthermore, $\Pb(X_j(k) = 1) = x_t(k)$ as node $j$ is chosen at random and $x_t(k)$ denotes the proportion of nodes with degree $t$ that are in state $\Xt$ at time $k$. Following identical arguments, we obtain
\begin{align}
\Eb[A_{ij}(k) Y_j(k)|v_i(k) \in \St\At; \pi] & = \left[\bar{t}\pi_{t,\Yt} + \bar{d}(1-\bar{t}\pi_{t,\Yt})\right]y_t(k).
\end{align}
Note that both expectations evaluate to an identical quantity for all nodes with a given degree $t$. Now, substituting the above in \eqref{eq:mfbound_siyr_1}, we obtain
\begin{align}
& \Pb(v_i(k+1) \in \Xt | v_i(k) \in \St\At; \pi) \nonumber 
\\ & \quad \quad \simeq \sum_{t \in \DD} nm_t \beta_x x_t(k) \left[\bar{t}\pi_{t,\Xt} \!+\! \bar{d}(1-\bar{t}\pi_{t,\Xt})\right] \!+\! nm_t \beta_y y_t(k) \left[\bar{t}\pi_{t,\Yt} \!+\! \bar{d}(1-\bar{t}\pi_{t,\Yt})\right] \nonumber
\\ & \quad \quad \simeq \sum_{t \in \DD} m_t \beta_x x_t(k) \left[t\pi_{t,\Xt} + d(1-\bar{t}\pi_{t,\Xt})\right] + m_t \beta_y y_t(k) \left[t\pi_{t,\Yt} +d(1-\bar{t}\pi_{t,\Yt})\right] \nonumber
\\ & \quad \quad =: \Rb_d(\pi, \bar{e}(k) | v_i(k) \in \St\At), \label{eq:risk_main_SA} 
\end{align}
where the first implication follows since the number of nodes with degree $t$ is $nm_t$ and $(n-1)m_d \simeq nm_d$ for large values of $n$. The last approximate equality also follows from $nt/(n-1) \simeq t$ for large values of $n$. The risk factor thus defined in \eqref{eq:risk_main_SA} is a function of the proportion of asymptomatic and symptomatic nodes and the activation probabilities in asymptomatic and symptomatic states for each degree $d \in \DD$. 

We now follow analogous arguments and approximate the probability of a susceptible node $v_i$ with degree $d$ becoming infected when it does not activate as
\begin{align}
& \Pb(v_i(k+1) \in \Xt | v_i(k) \in \St\Nt; \pi) \nonumber
\\ & \qquad \leq \sum_{j \neq i} \beta_x \Eb[A_{ij}(k)X_j(k)|v_i(k) \in \St\Nt; \pi] + \beta_y \Eb[A_{ij}(k) Y_j(k)|v_i(k) \in \St\Nt; \pi] \nonumber
\\ & \qquad \simeq \sum_{t \in \DD} \!\!nm_t \beta_x \Pb(A_{ij}(k) \!= 1|v_i(k) \in \St\Nt, X_j(k) \!= 1, d_j \!= t; \pi) \Pb(X_j(k) \!= 1| d_j \!= t) \nonumber
\\ & \qquad \quad + \sum_{t \in \DD} \!\!nm_t \beta_y \Pb(A_{ij}(k) \!= 1|v_i(k) \in \St\Nt, Y_j(k) \!= 1, d_j \!= t; \pi) \Pb(Y_j(k) \!= 1| d_j \!= t) \nonumber
\\ & \qquad = \sum_{t \in \DD} nm_t \beta_x x_t(k) \Pb(\Gamma^k_{j \to i} | \SXt^k_{ij}; \pi) + nm_t \beta_y y_t(k) \Pb(\Gamma^k_{j \to i} | \SYt^k_{ij}; \pi) \nonumber
\\ & \qquad \simeq \sum_{t \in \DD} \beta_x m_t x_t(k) t \pi_{t,\Xt} + \beta_y m_t y_t(k) t \pi_{t,\Yt} =: \Rb_d(\pi, \bar{e}(k) | v(k) \in \St\Nt). \label{eq:risk_main_SN}
\end{align}

Recall from Definition \ref{def:qredef_asiyr} that the equilibrium activation probability is a function of the difference in the utility of the nodes $\Delta u_{d,\Zt}(\pi,\bar{e}(k))$. Therefore, we compute
\begin{align*}
\Delta u_{d,\St}(\pi, \bar{e}(k)) &= u_{d,\St}(\At,\pi,\bar{e}(k)) - u_{d,\St}(\Nt,\pi,\bar{e}(k)) 
\\ & = 1 - L[\Rb_d(\pi, \bar{e}(k) | v(k) \in \St\At)-\Rb_d(\pi, \bar{e}(k) | v(k) \in \St\Nt)]
\\ & = 1 - L \left[\sum_{t \in \DD} m_t \beta_x x_t(k) d(1-\bar{t}\pi_{t,\Xt}) + m_t \beta_y y_t(k) d(1-\bar{t}\pi_{t,\Yt})\right].
\end{align*}

Thus, in order to decide whether to activate or not in the proposed framework, a node needs to be aware of the activation probabilities of asymptomatic and symptomatic nodes and the proportion of asymptomatic and symptomatic nodes for each degree present in the network. As a result, even when the nodes are aware of the required information, solving for the QRE requires \rev{finding a strategy profile $\pi$ that is a fixed point of the logit choice based best response mapping as stated in Definition \ref{def:qredef_asiyr},} which requires solving for the solution of $|\DD|$ coupled nonlinear equations. 

In order to reduce the information and computational burden, we further approximate $\Delta u_{d,\St}(\pi, \bar{e}(k))$. \rev{In particular, if the maximum degree $d_{\max}$ is much smaller compared to the total number of nodes $n$, which is a fairly reasonable assumption in a large network, we have $\bar{t}\pi_{t,\Zt} \leq \frac{d_{\max}}{n-1} \ll 1$. Then, assuming $1-\bar{t}\pi_{t,\Zt} \simeq 1$, we obtain}
\begin{align}\label{eq:diff_utility_asiyr}
\Delta u_{d,\St}(\pi, \bar{e}(k)) & \simeq 1 - L \left[\sum_{t \in \DD} m_t \beta_x x_t(k) d + m_t \beta_y y_t(k) d\right] \nonumber
\\ & = 1-Ld[\beta_x \bar{x}(k)+\beta_y \bar{y}(k)],
\end{align}
where $\bar{x}(k) := \sum_{t \in \DD} m_t x_t(k)$ and $\bar{y}(k) := \sum_{t \in \DD} m_t y_t(k)$ denote the expected fraction of (all) nodes that are asymptomatic and symptomatic at time $k$. The above expression is also independent of the activation decisions of other nodes. In the subsequent analysis and in our simulations, we use the approximation in \eqref{eq:diff_utility_asiyr}, and assume that the nodes are aware of $\bar{e}(k) = (\bar{x}(k),\bar{y}(k))$ at each time $k$. 

%%% remarks %%%%

\begin{remark}
\rev{To summarize, in this subsection, we define the risk factor as the approximate probability of becoming infected in the next time step at a given strategy profile $\pi$ and infection state $e(k)$. Our approximation assumes that $n \simeq n-1$ and relies on applications of the Weierstrass product inequality in \eqref{eq:mfbound_siyr_1} and in the first inequality in \eqref{eq:risk_main_SN}. Following earlier discussion in Remark \ref{remark:WPI_doerr}, the Weierstrass product inequality is tight if the summand is much smaller than $2$, i.e., when $R_d$ is small. Further, $R_d$ is an upper bound on the probability of infection, i.e., the nodes take a more conservative view of infection risk. Finally, the approximation in \eqref{eq:diff_utility_asiyr} is accurate when the maximum degree $d_{\max}$ is much smaller compared to the number of nodes $n$. This approximation renders $\Delta u_{d,\St}(\pi, \bar{e}(k))$ independent of the strategy profile $\pi$ and it depends only on the aggregate infection prevalence $(\bar{x}(k),\bar{y}(k))$.} \hfill \oprocendsymbol
\end{remark}

\subsection{Activation Probabilities at the Quantal Response Equilibrium} 

The following result characterizes the activation probabilities of the nodes at the QRE in different epidemic states of the A-SAIR model. 

\begin{proposition}\label{prop:qre_asiyr}
Consider the A-SAIR epidemic model with the utility functions of the nodes as defined in \eqref{eq:utilitydef_asiyr}. Let the nodes be aware of $\bar{e}(k) = (\bar{x}(k),\bar{y}(k))$ at time $k$ and let $\Delta u_{d,\St}(\pi, \bar{e}(k))$ be given as \eqref{eq:diff_utility_asiyr}. Then, the activation probabilities of the nodes at the quantal response equilibrium with parameter $\lambda$ are given by
\begin{align}\label{eq:qre_prob_asiyr}
\pi^\lambda_{d, \St,\NEt}(k) & = \pi^\lambda_{d, \Xt,\NEt}(k) = \frac{e^{\lambda - \lambda L d [\beta_x \bar{x}(k)+\beta_y \bar{y}(k)]}}{e^{\lambda- \lambda L d [\beta_x \bar{x}(k)+\beta_y \bar{y}(k)]}+1}, 
\\ \pi^\lambda_{d, \Rt,\NEt}(k) & = \frac{e^{\lambda}}{e^{\lambda}+1}, \qquad \pi^\lambda_{d, \Yt,\NEt}(k) = \frac{e^{\lambda(1-c)}}{e^{\lambda(1-c)}+1}, \qquad k \geq 0, d \in \DD. \hspace{10mm} \hfill \oprocendsymbol
\end{align}
\end{proposition}

The above result is a consequence of the fact that the utilities in symptomatic and recovered states do not depend on the degrees, epidemic states and the strategy profile; see \eqref{eq:utilitydef_asiyr}. Furthermore, in the susceptible and asymptomatic states, the difference in utilities only depend on the node degree and aggregate epidemic states $\bar{x}(k)$ and $\bar{y}(k)$ under the approximation \eqref{eq:diff_utility_asiyr}. Nodes use the aggregate epidemic states as {\it feedback} to determine the equilibrium activation probability $\pi^\lambda_{d, \St,\NEt}(k)$ in a decentralized manner under the logit choice framework. Since nodes with a larger degree experience a greater risk due to exposure to a larger number of nodes, the equilibrium activation probability is monotonically decreasing with respect to the node degree. 

\sloppy
When all nodes choose their activation probabilities according to their QRE strategies, the evolution of the epidemic states can be approximated via the DBMF approximation defined in Definition \ref{thm:asiyr_dbmf} with activation probabilities given by $\pi^\lambda_{d, \Zt,\NEt}(k), \Zt \in \{\St,\Xt,\Yt,\Rt\}$ as stated in Proposition \ref{prop:qre_asiyr}. 

\begin{remark}
Other notions of risk/cost of getting infected can be analyzed in the proposed framework and QRE strategies can be derived following Definition \ref{def:qredef_asiyr}. \hfill \oprocendsymbol
\end{remark}

\begin{remark}
\rev{In the proposed formulation, the agents evaluate their risk as the probability of becoming infected in the next time-step, and maximize their instantaneous payoffs in a myopic manner. While we do not consider agents who strategize in a look-ahead manner in this work, our assumption is a reasonable first approximation in the context of epidemics. Indeed, as observed during early stages of COVID-19, there is often considerable uncertainty involved about future trajectory of epidemics, which makes it difficult for an individual to plan actions over a longer time horizon. Nevertheless, analyzing the equilibrium behavior of agents who maximize their payoffs over a longer time horizon remains an interesting avenue for future research.} \hfill \oprocendsymbol
\end{remark}

We now discuss the implications of Proposition \ref{prop:qre_asiyr} for the A-SIR epidemic model.

\subsubsection{A-SIR Epidemic Model} In the conventional SIR epidemic model, there is no distinction between asymptomatic and symptomatic epidemic states. The activity-driven analogue of the SIR epidemic model, denoted by A-SIR epidemic model, is a special case of the A-SAIR epidemic with $\nu = 0$ and $Y_i(0) = 0$ for all nodes $v_i$. Thus, the individual-based and degree-based mean-field approximations for the A-SIR model are special cases of the approximations derived in Theorem \ref{thm:asiyr_main} and Definition \ref{thm:asiyr_dbmf} with $\nu = 0$ and $y_i(0) = 0$ or $y_d(0) = 0$ for all nodes/degrees.  

We assume that infected nodes are aware of being infected, i.e., they behave as symptomatic nodes in the A-SAIR model. As a consequence of the analysis in the previous subsection and Proposition \ref{prop:qre_asiyr}, we have
\begin{align}\label{eq:qre_asir_prob}
\pi^\lambda_{d, \Rt,\NEt}(k) & = \frac{e^{\lambda}}{e^{\lambda}+1}, \qquad 
\pi^\lambda_{d, \Xt,\NEt}(k) = \frac{e^{\lambda(1-c)}}{e^{\lambda(1-c)}+1},
\\ \pi^\lambda_{d, \St,\NEt}(k) & = \frac{e^{\lambda \Delta u_{d,\St}(\pi^\lambda_{\NEt}(k),\bar{x}(k))}}{e^{\lambda \Delta u_{d,\St}(\pi^\lambda_{\NEt}(k),\bar{x}(k))}+1}, \qquad d \in \DD,
\end{align}
where
\begin{align}
\Delta u_{d,\St}(\pi^\lambda_{\NEt}(k),\bar{x}(k)) & = 1 - Ld \sum_{t \in \DD}  \beta_x m_t x_t(k) (1-\bar{t}\pi^\lambda_{t,\Xt,\NEt}) \simeq 1 - Ld \beta_x \bar{x}(k),
\end{align}
following \eqref{eq:diff_utility_asiyr}. In the special case when all nodes have an identical degree $d$, then $\bar{x}(k) = x_d(k)$. In the following section, we analyze the activity-driven SIS epidemic model. In Section \ref{section:simulation}, detailed simulations are carried out for both epidemic models, and the impacts of cost parameters, $\lambda$ and degree distribution are investigated.
\section{A-SIS Epidemic under Game-Theoretic Activation}
\label{section:a-sis}

We now consider the activity-driven SIS (A-SIS) epidemic model under game-theoretic activation decisions by the nodes. An analogous model was first introduced in \cite{ogura2019optimal} with state-dependent activation probabilities, but without any strategic decision-making by the nodes. In the A-SIS epidemic model, each node remains in two possible states: susceptible ($\St$) and infected ($\Xt$). Given a network or contact pattern, the probabilistic state evolution is formally defined below which is analogous to Definition \ref{def:siyr}.

\begin{definition}\label{def:a-sis}
Let $\beta, \delta \in [0, 1]$ be constants pertaining to infection and recovery rates. The state of each node $v_i$ evolves as follows.
\begin{enumerate}
\item If $v_i(k) \in \St$, then $v_i(k+1) \in \Xt$ \rev{if it is infected by at least one neighbor. The probability of infection is $\beta$ for an infected neighbor. The infection caused by each neighbor is independent of the infections caused by others.} 
\item If $v_i(k) \in \Xt$, then $v_i(k+1) \in \St$ with probability $\delta$.
\end{enumerate}
The state remains unchanged otherwise. \hfill \oprocendsymbol
\end{definition}

\begin{figure}[tb]
\centering
\begin{tikzpicture}[font=\sffamily]

% Setup the style for the states
\tikzset{node style/.style={state, minimum width=1cm, line width=0.2mm, fill=orange!60!white}}
\tikzset{node style1/.style={state, minimum width=1cm, line width=0.2mm, fill=red!50!white}}
\tikzset{node style2/.style={state, minimum width=1cm, line width=0.2mm, fill=green!50!white}}
        % Draw the states
\node[node style] at (0, 0)     (St)     {$\St$};
\node[node style1] at (3, 0)   (At)     {$\Xt$};
        % Connect the states with arrows
        \draw[every loop,
              auto=right,
              line width=0.4mm,
              >=latex,
              draw=black,
              fill=black]
            %(St)     edge[bend right=20]            node {0.1} (Yt)
            (St)     edge[bend right=30, auto=left] node[above] {$\beta$} (At)
            (At)     edge[bend right=30] node[above] {$\delta$} (St);
    \end{tikzpicture}
\caption{\footnotesize Probabilistic evolution of states in the A-SIS epidemic model. Self-loops are omitted for better clarity. See Definition \ref{def:a-sis} for the formal definition.}
\label{fig:asis}
\end{figure}
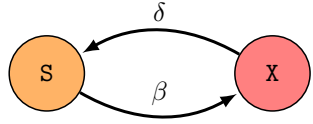

The state transitions are illustrated in Figure \ref{fig:asis}. The state-dependent evolution of the network or contact pattern is analogous to the definition in the A-SAIR epidemic setting (Definition \ref{def:act_asiyr}) with the only distinction being that nodes can be in one of two possible states and the epidemic parameters are $\beta$ and $\delta$ defined above. We denote the activation probabilities of a node $v_i$ at time $k$ when it is susceptible and infected by $\pi_{\St,i}(k)$ and $\pi_{\Xt,i}(k)$, respectively. 

In this section, we first derive the individual-based and degree-based mean-field approximations to the exact Markovian state evolution. We then derive the activation probabilities chosen by the nodes in different states at the QRE, followed by analyzing the closed-loop dynamics. 

\subsection{State Evolution and Mean-Field Approximations}\label{section:asis-mfa}

We follow analogous notations and arguments as the A-SAIR epidemic model from Section \ref{section:asiyr-mfa}. For a node $v_i$, we define a random variable $X_i(k)$ with $X_i(k) = 1$ if $v_i$ is infected at time $k$ and $X_i(k) = 0$ if susceptible. Under the A-SIS epidemic, the state of $v_i$ evolves as
\begin{align}\label{eq:asis_markov}
X_i(k+1) = (1-N_{\delta}) X_i(k) + (1-X_i(k)) \big[1 - \underset{j \neq i}{\Pi} \left[ 1 - A_{ij}(k) X_j(k)N_{\beta} \right]\big],
\end{align}
where $A_{ij}(k), N_{\delta}$ and $N_{\beta}$ are as defined in Section \ref{section:asiyr-mfa}. We define $x_i(k) := \Pb(X_i(k)=1) = \Eb[X_i(k) = 1]$ as before. Note that the epidemic states follow a Markov process with a $2^n \times 2^n$ transition probability matrix. The following result derives an individual-based mean-field approximation.  

\begin{theorem}\label{thm:asis}
Consider the A-SIS epidemic model where node $v_i$ connects to $d_i$ other nodes upon activation. Let $\bar{d}_i = d_i/(n-1)$. Then,
\begin{align}\label{eq:mfa_linear_dynamics_asis}
x_i(k\!+\!1) & \simeq \!(1\!-\!\delta)x_i(k) \!+ \!(1\!-x_i(k))\! \sum_{j \neq i} \!\beta \big[1\!-\!(1\!-\pi_{\St,i}(k) \bar{d}_i)(1 \!- \pi_{\Xt,j}(k) \bar{d}_j)\big]  x_j(k) 
\end{align}
for all nodes $v_i$ and $k \geq 0$. \hfill \oprocendsymbol
\end{theorem}

The proof is analogous to the proof of Theorem \ref{thm:asiyr_main}, \rev{and is presented in Appendix \ref{sec:appendix_proofs}. It relies on an application of the Weierstrass product inequality and assuming that the infection states of two nodes $v_i$ and $v_j$ are independent (as before).} We now define the DBMF approximation of the A-SIS epidemic model.

\begin{definition}\label{def:asis_dbmf}
Consider the A-SIS epidemic model defined in Definition \ref{def:a-sis}. Let all nodes with degree $d \in \DD$ choose an identical (state-dependent) activation probability denoted by $\pi_{d,\Zt}(k)$ with $\Zt \in \{\St,\Xt\}$. Then, the proportion of nodes with degree $d$ that are infected evolves as
\begin{equation}\label{eq:dbmf_asis}
 x_d(k\!+\!1) \!= \!(1\!-\!\delta)x_d(k) \!+ \!(1\!-\!x_d(k))\! \sum_{t \in \DD} \!\!nm_t\beta \big[1\!-\!(1\!-\!\pi_{d,\St}(k) \bar{d})(1 \!-\! \pi_{t,\Xt}(k) \bar{t})\big] x_t(k) 
\end{equation}
where $\bar{d} = d/(n-1)$ and $\bar{t} = t/(n-1)$. \hfill \oprocendsymbol
\end{definition}

We arrived at the above definition from the individual based mean-field approximation in Theorem \ref{thm:asis}, by noting that all nodes with degree $t$ contribute an identical quantity to the summation and the number of nodes with degree $t$ is $nm_t$. 

The above approximations are analogous to the deterministic evolution and the N-Intertwined Mean-Field Approximation (NIMFA) of the classical SIS epidemic \cite{van2011n,pare2018analysis}.  If the activation probabilities remain constant as the epidemic evolves, then we have the following result on their equilibrium behavior. 

\begin{proposition}\label{prop:sis_endemic_rho}
Consider the DBMF approximation of the A-SIS epidemic in Definition \ref{def:asis_dbmf} with $\delta, \beta \in (0,1]$. Let the activation probabilities be time-invariant and denoted by $\pi_{d,\St}, \pi_{d,\Xt} \in (0,1]$ for $d \in \DD$. Let $\beta(d_{\avg} + d_{\max}) \leq 1$.  Consider the matrix $\FF \in \Rb^{|\DD| \times |\DD|}$ with entries
\begin{align*}
[\FF]_{d,t} = \begin{cases}
(1-\delta) + nm_d \beta \big[1-(1-\pi_{d,\St} \bar{d})(1 - \pi_{d,\Xt} \bar{d})\big], \qquad & \text{if $t = d$,}
\\ nm_t \beta \big[1-(1-\pi_{d,\St} \bar{d})(1 - \pi_{t,\Xt} \bar{t})\big], \qquad & \text{otherwise.}
\end{cases}
\end{align*}
Let $\rho(\FF)$ denote the spectral radius of $\FF$. Then, $[0,1]^{|\DD|}$ is an invariant set for \eqref{eq:dbmf_asis},
\begin{enumerate}
\item if $\rho(\FF) \leq 1$, then $x = \mathbf{0} \in \Rb^{|\DD|}$ is the unique (disease free) equilibrium of \eqref{eq:dbmf_asis} which is asymptotically stable with domain of attraction $[0,1]^{|\DD|}$, and
\item if $\rho(\FF) > 1$, there are two equilibria, $\mathbf{0}$ and $x^*$ with $x^*_d > 0, \forall d \in \DD$. If in addition $\delta + \beta(d_{\avg} + d_{\max}) \leq 1$, then the endemic equilibrium $x^*$ is asymptotically stable with domain of attraction $[0,1]^{|\DD|} \setminus \{\mathbf{0}\}$. \hfill \oprocendsymbol
\end{enumerate}
\end{proposition}

\rev{The proof follows by adapting analogous results obtained in \cite{pare2018analysis,liu2020stability} for the discrete-time (classical) SIS epidemic dynamic to the dynamic under the DBMF approximation. The main result from \cite{pare2018analysis,liu2020stability} is presented in Appendix \ref{sec:appendix_prior} and the proof of Proposition \ref{prop:sis_endemic_rho} is presented in Appendix \ref{sec:appendix_proofs}.} While we have stated the above result for the DBMF approximation, an analogous result holds for the individual-based MFA derived in Theorem \ref{thm:asis} as well. We omit the details in the interest of space.

%%%%%%%%%%%%%%%%%%%%%%%%%%%%%%%
%%%%%%%%%%%%%%%%%%%%%%%%%%%%%%%

\subsection{Game-Theoretic Choice of Activation Probabilities}

We now consider the activation game for the SIS epidemic model. Analogous to the setting studied for the A-SAIR epidemic in Section \ref{section:activation_game}, we assume that all nodes with a given degree $d$ choose an identical activation probability. We now formally define the utilities of the nodes building upon the setting in Section \ref{section:activation_game}. 

We assume that all nodes receive benefit $1$ upon activation. However, an infected node incurs a cost $c \in \Rb_{+}$ if it activates. Susceptible nodes weigh the risk of becoming infected in the next time step by a loss $L \in \Rb_{+}$. Formally, let $\pi_{d,\Zt} \in [0,1]$ denotes the activation probability of a node with degree $d$ in epidemic state $\Zt \in \{\St,\Xt\}$. The strategy profile of the nodes is denoted by $\pi := \{\pi_{d,\St}, \pi_{d,\Xt}\}_{d \in \DD} \subseteq [0,1]^{2|\DD|}$. The global epidemic state at time $k$ is denoted by $e(k) := \{x_d(k)\}_{d \in \DD}$, i.e, the proportion of infected nodes in each subpopulation. Note that since there are only two epidemic states, we have $s_d(k) = 1 - x_d(k)$. Let $\bar{e}(k)$ denote the information available to the nodes regarding the epidemic state. The utilities of a node $v$ with degree $d \in \DD$ in different epidemic states at time $k$ are defined as:
\begin{subequations}\label{eq:utilitydef_asis}
\begin{align}
& u_{d,\St}(\At,\pi,\bar{e}(k)) = 1 - L\Rb_d(\pi, \bar{e}(k) | v(k) \in \St\At),
\\ & u_{d,\St}(\Nt,\pi, \bar{e}(k)) = - L\Rb_d(\pi, \bar{e}(k) | v(k) \in \St\Nt),
\\ & u_{d,\Xt}(\At,\pi, \bar{e}(k)) = 1 - c, \qquad \qquad \quad u_{d,\Xt}(\Nt,\pi, \bar{e}(k)) = 0,
\end{align}
\end{subequations}
where $\Rb_d(\pi, \bar{e}(k) | v(k) \in \St\At)$ and $\Rb_d(\pi, \bar{e}(k) | v(k) \in \St\Nt)$ denote the respective risk factors for a susceptible node of degree $d$ under strategy profile $\pi$ if it activates (denoted by $\St\At$) and does not activate (denoted by $\St\Nt$).

We define the risk factors as approximations of the probability with which a susceptible node becomes infected in the next time step following the approach in Section \ref{section:riskfactor_mfa}. We specifically note that the analysis in Section \ref{section:riskfactor_mfa} continues to hold for the A-SIS epidemic model when we exclude the symptomatically infected state. Therefore, building upon \eqref{eq:risk_main_SA} and \eqref{eq:risk_main_SN}, we define
\begin{align*}
\Rb_d(\pi, \bar{e}(k) | v(k) \in \St\At) & := \sum_{t \in \DD} \beta m_t x_t(k) \left[t\pi_{t,\Xt} + d(1-\bar{t}\pi_{t,\Xt})\right],
\\ \Rb_d(\pi, \bar{e}(k) | v(k) \in \St\Nt) & := \sum_{t \in \DD} \beta m_t x_t(k) t \pi_{t,\Xt}, 
\\ \Delta u_{d,\St}(\pi, \bar{e}(k)) & = 1 - L \left[\sum_{t \in \DD} m_t \beta x_t(k) d(1-\bar{t}\pi_{t,\Xt})\right] \simeq 1-Ld\beta \bar{x}(k),
\end{align*}
where $\bar{x}(k)$ is the proportion of nodes that are infected at time $k$ (across all subpopulations) and we have assumed $d_{\max} \ll n$. Accordingly, the activation probabilities at the QRE with parameter $\lambda$ are given by
\begin{align}\label{eq:qre_asis_actprob}
\pi^\lambda_{d, \St,\NEt}(k) & := \frac{e^{\lambda - \lambda L d \beta \bar{x}(k)}}{e^{\lambda - \lambda L d \beta \bar{x}(k)}+1}, \qquad 
\pi^\lambda_{d, \Xt,\NEt}(k) = \frac{e^{\lambda(1-c)}}{e^{\lambda(1-c)}+1} =: \pi^\lambda_c, \qquad d \in \DD.
\end{align}

In other words, nodes use the information regarding the proportion of infected nodes in the network as feedback and adjust their the activation probabilities to maximize their individual utility in the logistic quantal response choice model. When $\lambda = 0$, the activation probabilities are $0.5$. For larger values of $\lambda$, nodes choose the action that leads to a higher utility with a greater probability. 

When a node is already infected, the activation decision depends on the penalty parameter $c$. When $c > 1$, the penalty is larger than the benefit from activation, and as a result, $\pi^\lambda_c \leq 0.5$ with $\pi^\lambda_c \to 0$ as $\lambda \to \infty$. For susceptible nodes, the higher the degree and the prevalence of the epidemic, the smaller is the activation probability at the QRE. For sufficiently high $\lambda$, if $dL\beta\bar{x}(k)<1$, then nodes with degree $d$ activate with high probability; otherwise the equilibrium activation probability is negligible. 

\subsection{Analysis of Closed-Loop Dynamics}\label{section:sis_closedloop}

We now analyze the evolution of the epidemic states when nodes choose their activation probabilities according to \eqref{eq:qre_asis_actprob} in the DBMF approximation framework. Following Definition \ref{def:asis_dbmf}, the evolution of the proportion of infected nodes with degree $d$ is given by
\begin{align}\label{eq:asis_dbmf_closed}
x_d(k+1) & = (1-\delta) x_d(k) + (1 - x_d(k)) \nonumber
\\ & \qquad \cdot \sum_{t \in \DD} \!nm_t \beta \left[1-\Big(1-\frac{e^{\lambda - \lambda L d \beta \bar{x}(k)}}{e^{\lambda - \lambda L d \beta \bar{x}(k)}+1} \bar{d}\Big)(1 \!- \pi^\lambda_c \bar{t})\right] x_t(k).
\end{align}

Note that under QRE activation probabilities, the R.H.S. of $x_d(k+1)$ is nonlinear as $\pi^\lambda_{d,\St,\NEt}(k)$ depends on $\bar{x}(k)$. Therefore, the equilibrium characterization from Proposition \ref{prop:sis_endemic_rho} is not applicable here. We now exploit the structure of \eqref{eq:asis_dbmf_closed} to obtain further insights on the closed-loop dynamics for specific parameter regimes.  

\subsubsection{Rationality parameter $\lambda = 0$}
In this case, nodes choose the activation probability to be $0.5$ irrespective of their degree and epidemic state. Furthermore, the dynamics is independent of cost parameters ($c$ and $L$) as the activation decisions are not strategic. Therefore, the behavior of the closed-loop dynamics \eqref{eq:asis_dbmf_closed} is in accordance with Proposition \ref{prop:sis_endemic_rho} with $\pi_{d, \St} = \pi_{d, \Xt} = 0.5$ for all $d \in \DD$.

\subsubsection{Rationality parameter $\lambda \to \infty$, Cost $c > 1$}
\rev{Recall that as $\lambda \to \infty$, individuals have a high degree of rationality and choose the action with higher payoff with probability $1$. When $c > 1$, the cost of activation exceeds the benefit for infected nodes, and therefore, $\pi^{\lambda}_c \to 0$ as $\lambda \to \infty$. Similarly, as $\lambda \to \infty$, the activation probability of susceptible nodes is given by
\begin{equation}\label{eq:act_prob_lbminf}
\lim_{\lambda \to \infty} \pi^\lambda_{d,\St,\NEt}(k) = \lim_{\lambda \to \infty} \frac{e^{\lambda - \lambda L d \beta \bar{x}(k)}}{e^{\lambda - \lambda L d \beta \bar{x}(k)}+1} = 
\begin{cases}
0, \quad \text{if } \bar{x}(k) > \frac{1}{L\beta d}, 
\\ \frac{1}{2}, \quad \text{if } \bar{x}(k) = \frac{1}{L\beta d},
\\ 1, \quad \text{if } \bar{x}(k) < \frac{1}{L\beta d}.
\end{cases}
\end{equation}}
\rev{Consequently, the dynamics in \eqref{eq:asis_dbmf_closed} is given by
\begin{align}\label{eq:cl_sis_rational_prior}
x_d(k+1) \!= \begin{cases}
(1\!-\delta) x_d(k), \qquad & \text{if } \bar{x}(k) \in (\frac{1}{L\beta d},1],
\\ (1\!-\delta) x_d(k) \!+ (1-x_d(k)) \beta d \sum_{t \in \DD} m_t x_t(k), & \text{if } \bar{x}(k) \in [0, \frac{1}{L\beta d})
\\ (1\!-\delta) x_d(k) \!+ (1-x_d(k)) \frac{\beta d}{2} \sum_{t \in \DD} m_t x_t(k) , & \text{if } \bar{x}(k) = \frac{1}{L\beta d}.
\end{cases} 
\end{align}}
\rev{Thus, the closed-loop dynamics admits the form of a variable structure system which switches between different modes depending on the value of $\bar{x}(k) = \sum_{t \in \DD} m_t x_t(k)$. For each degree $d$, activation probability is $0$ if the proportion of infected nodes exceeds the threshold $\frac{1}{L\beta d}$, and $1$ otherwise. In order to avoid intricate corner cases, we focus on the following dynamics in the remainder of this section:}
\rev{\begin{align}\label{eq:cl_sis_rational}
x_d(k+1) \!= \begin{cases}
(1\!-\delta) x_d(k), \qquad & \text{if } \bar{x}(k) \in (\frac{1}{L\beta d},1],
\\ (1\!-\delta) x_d(k) \!+ (1-x_d(k)) \beta d \sum_{t \in \DD} m_t x_t(k), & \text{if } \bar{x}(k) \in [0, \frac{1}{L\beta d}]. 
\end{cases} 
\end{align}}

\begin{remark}\label{remark:NPI_inter}
\rev{The dynamics stated in \eqref{eq:cl_sis_rational} is also of independent interest as it captures a class of non-pharmaceutical interventions where interactions or gatherings beyond a certain size are restricted as a function of infection prevalence.} \hfill \oprocendsymbol
\end{remark}

\rev{We now investigate the equilibria and stability of \eqref{eq:cl_sis_rational} by exploiting its structure. We first consider the more tractable case when all nodes have degree $d$.}

\begin{proposition}\label{prop:sis_cl_homo}
\rev{Consider the dynamics \eqref{eq:cl_sis_rational} where all nodes have degree $d$, i.e., $\DD = \{d\}$ and $m_d = 1$. Let $\beta, \delta \in (0,1]$, $\beta d \leq 1$ and $x_d(0) \in [0,1]$. Then,
\begin{enumerate}
\item if $\delta \geq \beta d$, then $x_d = 0$ is the unique equilibrium, and is asymptotically stable with domain of attraction $[0,1]$, 
\item if $\delta < \beta d$ and $x^*_d := 1 - \frac{\delta}{\beta d} \leq \frac{1}{L\beta d}$, then there are two equilibrium points, $0$ and $x^*_d$; furthermore, $x^*_d$ is stable with domain of attraction $(0,1]$, 
\item if $\delta < \beta d$ and $x^*_d >  \frac{1}{L\beta d}$, then $\frac{1}{L\beta d}$ acts as a sliding surface for $x_d(k)$ for $x_d(0) \in (0,1]$. \hfill \oprocendsymbol
\end{enumerate}}
\end{proposition}

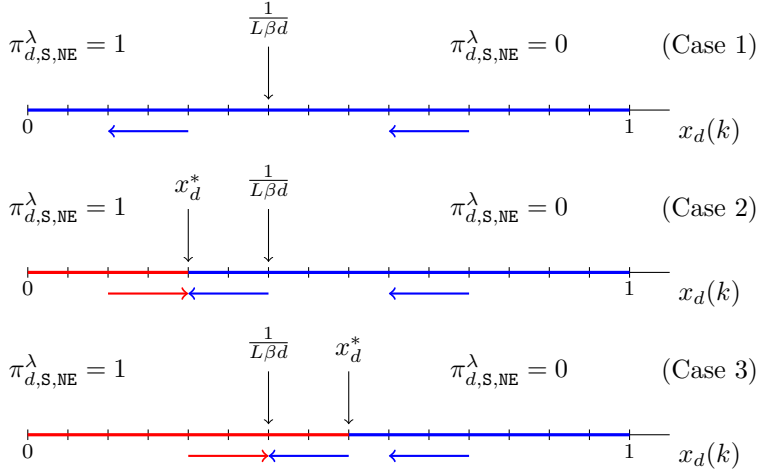
\begin{figure}
\centering
\begin{tikzpicture}[scale = 8] 
\draw (0,0) -- (16/15,0); 
\foreach \x in {0,1/15,2/15,1/5,4/15,5/15,2/5,7/15,8/15,3/5,10/15,11/15,4/5,13/15,14/15,1} \draw[shift={(\x,0)}] (0pt,2/8pt) -- (0pt,-2/8pt); 
\foreach \x in {0,1} 
\draw[shift={(\x,0)}] node[below] {\footnotesize $\x$}; 
\draw[very thick, blue] (6/15,0) -- (1,0); 
\draw[very thick, blue] (0,0pt) -- (6/15,0pt); 
\draw[->] (6/15, 3 pt)node[above]{$\frac{1}{L\beta d}$} -- (6/15,1/2 pt); 
\draw[->, thick, blue] (4/15, -1 pt) -- (2/15,-1 pt); 
\draw[->, thick, blue] (11/15, -1 pt) -- (9/15,-1 pt); 
\node at (17/15, 3 pt) {(Case 1)}; 
\node at (12/15, 3 pt) {$\pi^\lambda_{d,\St,\NEt} = 0$}; 
\node at (17/15, -1 pt) {$x_d(k)$}; 
\node at (1/15, 3 pt) {$\pi^\lambda_{d,\St,\NEt} = 1$}; 
\end{tikzpicture}
\begin{tikzpicture}[scale = 8] 
\draw (0,0) -- (16/15,0); 
\foreach \x in {0,1/15,2/15,1/5,4/15,5/15,2/5,7/15,8/15,3/5,10/15,11/15,4/5,13/15,14/15,1} \draw[shift={(\x,0)}] (0pt,2/8pt) -- (0pt,-2/8pt); 
\foreach \x in {0,1} 
\draw[shift={(\x,0)}] node[below] {\footnotesize $\x$}; 
\draw[very thick, blue] (4/15,0) -- (1,0); 
\draw[very thick, red] (0,0pt) -- (4/15,0pt); 
\draw[->] (6/15, 3 pt)node[above]{$\frac{1}{L\beta d}$} -- (6/15,1/2 pt); 
\draw[->] (4/15, 3 pt)node[above]{$x^*_d$} -- (4/15,1/2 pt); 
\node at (12/15, 3 pt) {$\pi^\lambda_{d,\St,\NEt} = 0$}; 
\node at (17/15, -1 pt) {$x_d(k)$}; 
\node at (1/15, 3 pt) {$\pi^\lambda_{d,\St,\NEt} = 1$};
\draw[->, thick, red] (2/15, -1 pt) -- (4/15,-1 pt); 
\draw[->, thick, blue] (6/15, -1 pt) -- (4/15,-1 pt);
\draw[->, thick, blue] (11/15, -1 pt) -- (9/15,-1 pt); 
\node at (17/15, 3 pt) {(Case 2)};  
\end{tikzpicture}
\begin{tikzpicture}[scale = 8] 
\draw (0,0) -- (16/15,0); 
\foreach \x in {0,1/15,2/15,1/5,4/15,5/15,2/5,7/15,8/15,3/5,10/15,11/15,4/5,13/15,14/15,1} \draw[shift={(\x,0)}] (0pt,2/8pt) -- (0pt,-2/8pt); 
\foreach \x in {0,1} 
\draw[shift={(\x,0)}] node[below] {\footnotesize $\x$}; 
\draw[very thick, blue] (8/15,0) -- (1,0); 
\draw[very thick, red] (0,0pt) -- (8/15,0pt); 
\draw[->] (6/15, 3 pt)node[above]{$\frac{1}{L\beta d}$} -- (6/15,1/2 pt); 
\node at (12/15, 3 pt) {$\pi^\lambda_{d,\St,\NEt} = 0$}; 
\node at (17/15, -1 pt) {$x_d(k)$}; 
\node at (1/15, 3 pt) {$\pi^\lambda_{d,\St,\NEt} = 1$};
\draw[->] (8/15, 3 pt)node[above]{$x^*_d$} -- (8/15,1/2 pt); 
\draw[->, thick, red] (4/15, -1 pt) -- (6/15,-1 pt); 
\draw[->, thick, blue] (8/15, -1 pt) -- (6/15,-1 pt);
\draw[->, thick, blue] (11/15, -1 pt) -- (9/15,-1 pt); 
\node at (17/15, 3 pt) {(Case 3)};   
\end{tikzpicture}
\caption{\small{\rev{Evolution of the infected proportion $x_d(k)$ under game-theoretic activation when all nodes degree $d$ in the regime $\lambda \to \infty, c > 1$ for the three cases stated in Proposition \ref{prop:sis_cl_homo}. The arrows indicate the direction along which $x_d(k)$ evolves under the dynamic \eqref{eq:cl_sis_rational}; red (blue) arrow indicates the regime where $x_d(k)$ increases (decreases). The axis is marked with red (blue) to denote regimes where $x_d(k)$ increases (decreases) when $\pi_{d,\St} = 1$.}}}
\label{fig:asis_dbmf_cl}
\end{figure}

The proof is presented in Appendix \ref{sec:appendix_sishet_proofs}. Figure \ref{fig:asis_dbmf_cl} provides a visual representation of each of the three cases. The result shows that if the loss upon infection $L$ is sufficiently small, then the nodes continue to activate with probability $1$ and the infected proportion converges to $x^*_d$ or $0$ depending on the infection and recovery rates. For sufficiently large $L$, nodes do not activate for $x_d(k) > \frac{1}{L\beta d}$. Therefore, when $x^*_d > \frac{1}{L\beta d}$, $x^*_d$ is no longer an equilibrium point of the closed loop dynamics; instead $x_d(k)$ oscillates around the threshold $\frac{1}{L\beta d}$. Therefore, by increasing $L$, the infected proportion can be maintained at $\frac{1}{L\beta d}$ under decentralized decision-making instead of $x^*_d$ which is the endemic state when susceptible nodes activate with probability $1$. Our simulations in Section \ref{sec:sim_asis} illustrate this phenomenon.

\rev{Before we analyze the equilibria of the closed-loop dynamics \eqref{eq:cl_sis_rational} with heterogeneous node degrees, we introduce some notation and intermediate results. Consider the dynamics
\begin{equation}\label{eq:allact_sis_dyn}
x_d(k+1) = (1-\delta)x_d(k) + (1-x_d(k)) \beta d \sum_{t \in \DD'} m_t x_t(k), \qquad d \in \DD' \subseteq \DD,
\end{equation}
which is akin to the DBMF approximation of the epidemic with degrees $d \in \DD'$, $\pi_{d,\St} = 1$ and $\pi_{d,\Xt} = 0$. We now state the following result.}

\begin{proposition}\label{cor:dt_allact}
\rev{Suppose $\delta, \beta \in (0,1]$, $\beta d (\sum_{t \in \DD'} m_t) \leq 1$ for all $d \in \DD'$. Then, if $\delta \geq \beta \sum_{t \in \DD'} m_t t$, then the disease-free state $x = \mathbf{0}$ is the unique equilibrium of \eqref{eq:allact_sis_dyn}, and is asymptotically stable with domain of attraction $[0,1]^{\DD'}$. If $\delta < \beta \sum_{t \in \DD'} m_t t$, then there exist two equilibria, $\mathbf{0}$ and an endemic equilibrium, $x^*(\DD')$ with $x^*_t(\DD') > 0, \forall t \in \DD'$. If in addition $\delta + \beta d (\sum_{t \in \DD'} m_t) \leq 1$, then $x^*(\DD')$ is asymptotically stable with domain of attraction $[0,1]^{|\DD'|} \setminus \{\mathbf{0}\}$.} \hfill \oprocendsymbol
\end{proposition}  

\rev{The proof is presented in Appendix \ref{sec:appendix_sishet_proofs}. We do not require $\sum_{t \in \DD'} m_t$ to be $1$ for the above result. Let the mass of infected nodes (across all degrees) at the endemic equilibrium of \eqref{eq:allact_sis_dyn} be denoted by $\bar{x}^*(\DD') = \sum_{t \in \DD'} m_t x^*_t(\DD')$. We now state the following useful result with proof presented in Appendix \ref{sec:appendix_sishet_proofs}.} 

\begin{proposition}\label{prop:endemic_monotone}
\rev{Let $\DD$ be a set of degrees with the mass of nodes with degree $d$ being $m_d$, and let $\delta, \beta \in (0,1]$, $\beta d (\sum_{t \in \DD} m_t) \leq 1$ for all $d \in \DD$. Let $\DD' \subseteq \DD$ with the mass of nodes with degree $d \in \DD'$ being $m_d$ as before. If the epidemic dynamics \eqref{eq:allact_sis_dyn} has an endemic state for degrees $\DD'$, then there exists an endemic state for the dynamics \eqref{eq:allact_sis_dyn} with degrees being $\DD$, and $\bar{x}^*(\DD') \leq \bar{x}^*(\DD)$.} \hfill \oprocendsymbol
\end{proposition}

We are leverage the above results to analyze the closed-loop dynamics \eqref{eq:cl_sis_rational} when node degrees are heterogeneous. Without loss of generality, let $\DD = \{d_1, d_2, \ldots, d_p\}$ with $d_1 < d_2 < \ldots < d_p$. Let $\DD_{j} := \{d_1, \ldots, d_{j}\}$, and let $\bar{x}^{**}(\DD_{j})$ be the proportion of infected nodes at the stable equilibrium of \eqref{eq:allact_sis_dyn} with degrees in $\DD_j$, i.e., 
\begin{equation}
\bar{x}^{**}(\DD_{j}) := \begin{cases} 0, \qquad & \text{if } \delta \geq \beta \sum_{d \in \DD_j} m_{d} d, 
\\ \bar{x}^{*}(\DD_{j}) = \sum_{d \in \DD_j} m_{d} x^*_{d}(\DD_j), \qquad & \text{if } \delta < \beta \sum_{d \in \DD_j} m_{d} d.
\end{cases}
\end{equation}

We now state the following main result.

\begin{theorem}\label{thm:dt_cl_sis_main}
\rev{Consider the DBMF approximation of the closed-loop SIS epidemic dynamics \eqref{eq:cl_sis_rational}. Let $\beta, \delta \in (0,1]$ and $\beta d_p \leq 1$. Then, 
\begin{enumerate}
\item if $\delta \geq \beta \sum_{d \in \DD} d m_d$, then $x = \mathbf{0} \in \Rb^{p}$ is the unique equilibrium of the dynamics, and it is asymptotically stable with region of attraction $[0,1]^p$. 
\item Otherwise, let $d_r \in \{d_1, \ldots, d_p\}$ be the largest degree with $\bar{x}^{**}(\DD_{r}) \leq \frac{1}{L\beta d_r}$. Then, there exists a nonzero endemic equilibrium $\mathbf{x}^*_{cl}$ if and only if $d_r = d_p$ or $d_r < d_p$ with $\bar{x}^{**}(\DD_{r}) \in \big(\frac{1}{L\beta d_{r+1}},\frac{1}{L\beta d_r}]$. Furthermore, $\mathbf{x}^*_{cl}$ is given by
\begin{equation}\label{eq:cl_endemic_final}
\mathbf{x}^*_{d,cl} := \begin{cases} x^*_{d}(\DD_r), \qquad & d \leq d_r, 
\\ 0, \qquad & d > d_r.
\end{cases} 
\end{equation} \hfill \oprocendsymbol
\end{enumerate}}
\end{theorem}

The proof of the above theorem is presented in Appendix \ref{sec:appendix_sishet_proofs}. 

\begin{remark}
We conclude this section with two open questions. 
\begin{itemize} 
\item While the above result characterizes the existence of an endemic equilibrium, we conjecture that $\mathbf{x}^*_{cl}$ is asymptotically stable with domain of attraction $[0,1]^p \setminus \{\mathbf{0}\}$, possibly under mild additional assumptions. 
\item If $\bar{x}^{**}(\DD_{r}) \leq \frac{1}{L\beta d_{r+1}}$, we conjecture that $\frac{1}{L\beta d_{r+1}}$ acts as a sliding surface for the closed-loop dynamics. We hope for further explorations along these lines in follow up work. \hfill \oprocendsymbol
\end{itemize}
\end{remark}

%%%%%%%%%%%%%%%%%%%%%%%%%%%%%%%
%%%%%%%%%%%%%%%%%%%%%%%%%%%%%%%
%%%%%%%%%%%%%%%%%%%%%%%%%%%%%%%
\section{Simulation Results}
\label{section:simulation}

In this section, we illustrate the evolution of the epidemic states under game-theoretic activation via extensive simulations. 

\subsection{A-SIR epidemic with homogeneous node degrees} 

In order to isolate the impacts of game-theoretic activation decisions and cost parameters ($\lambda, c$ and $L$) from the impacts of asymptomatic carriers and heterogeneous node degrees, we first consider the special case of A-SIR epidemic with homogeneous node degrees. 

\begin{figure*}[tb]
	\centering
	\includegraphics[scale=0.5]{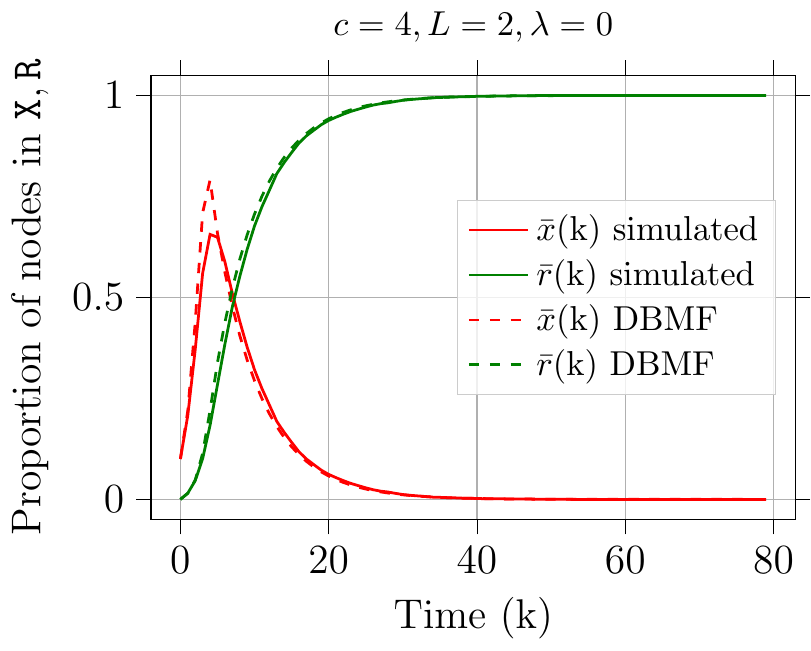}
	\includegraphics[scale=0.5]{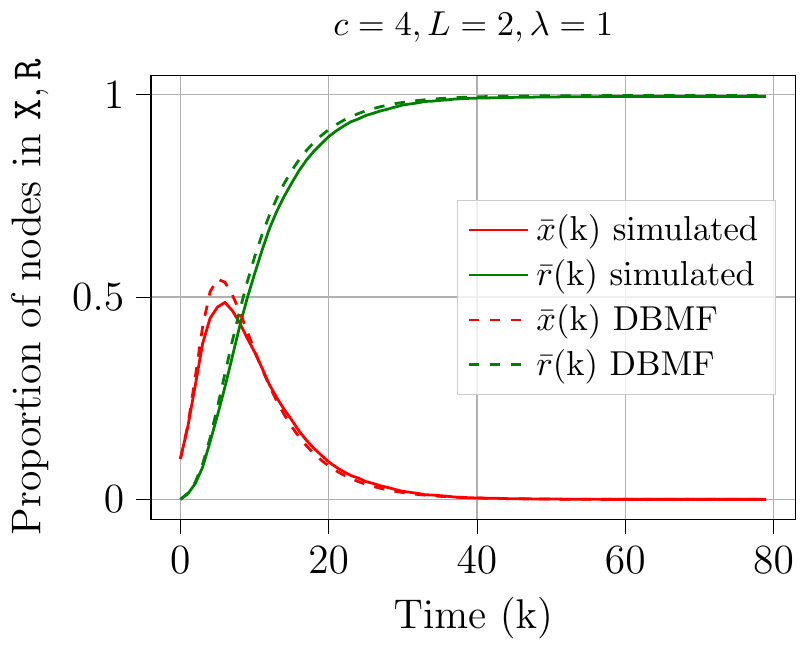}
	\includegraphics[scale=0.5]{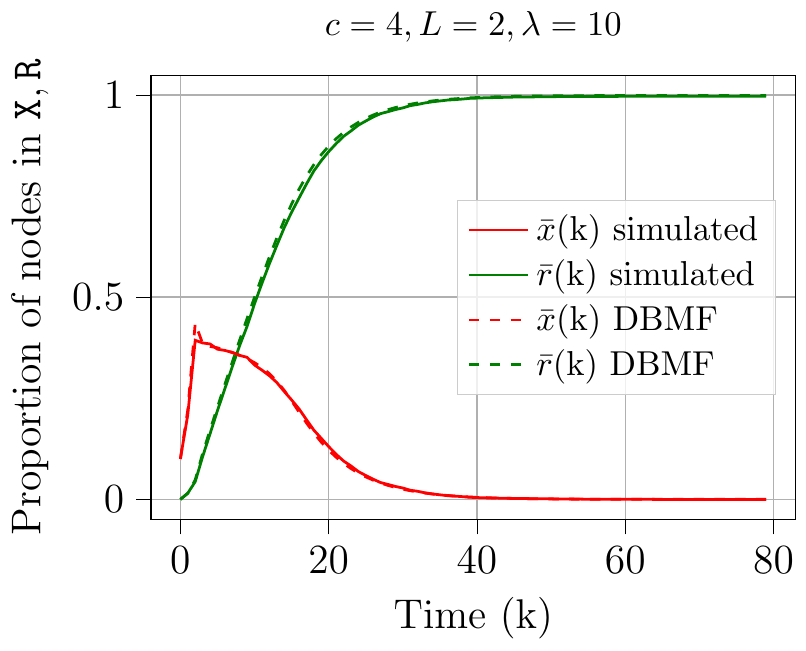} \\[2mm]
	\includegraphics[scale=0.5]{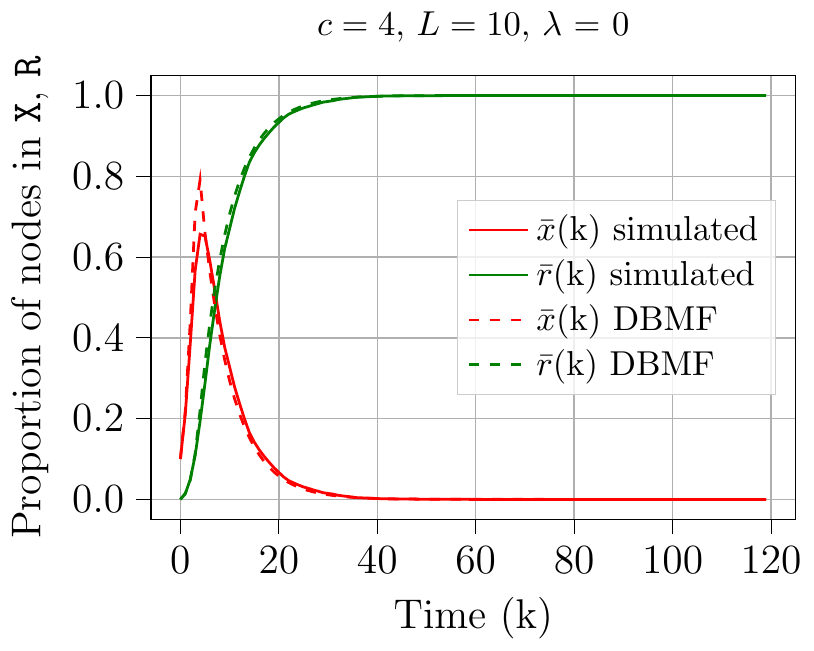}
	\includegraphics[scale=0.5]{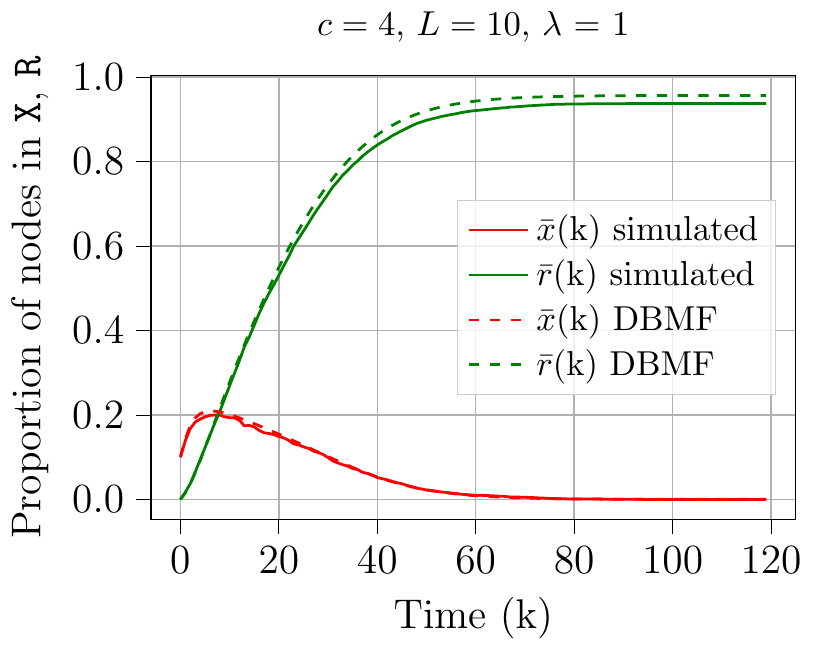}
	\includegraphics[scale=0.5]{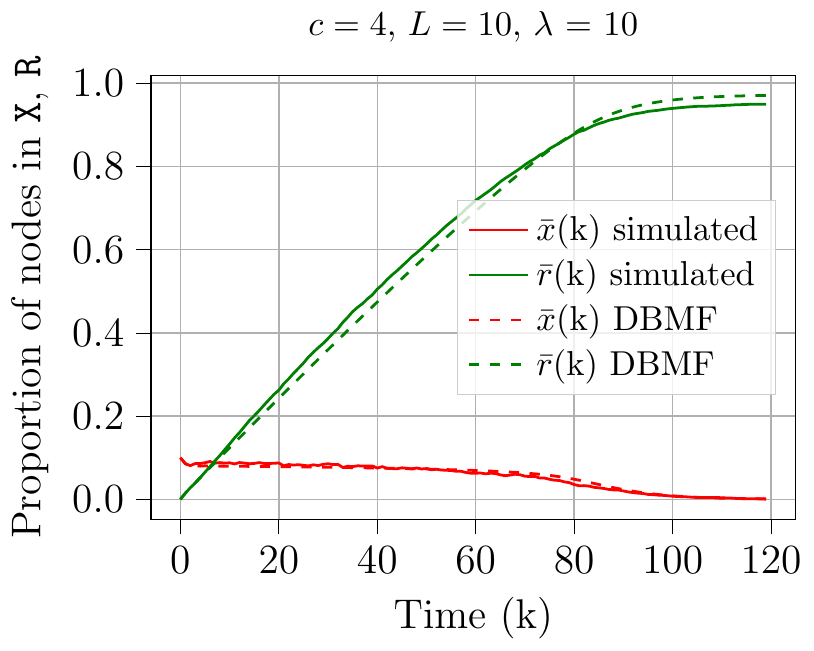} \\[2mm]
	\includegraphics[scale=0.5]{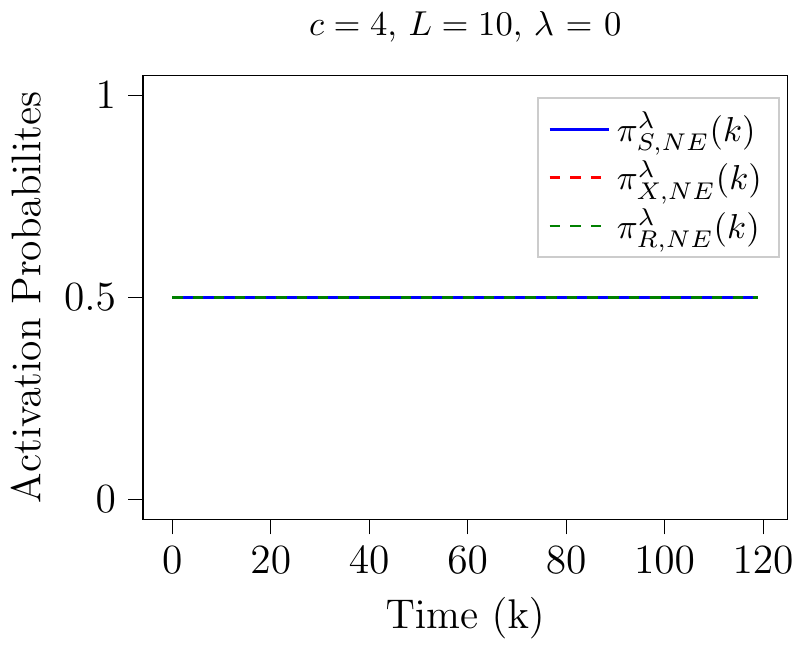}
	\includegraphics[scale=0.5]{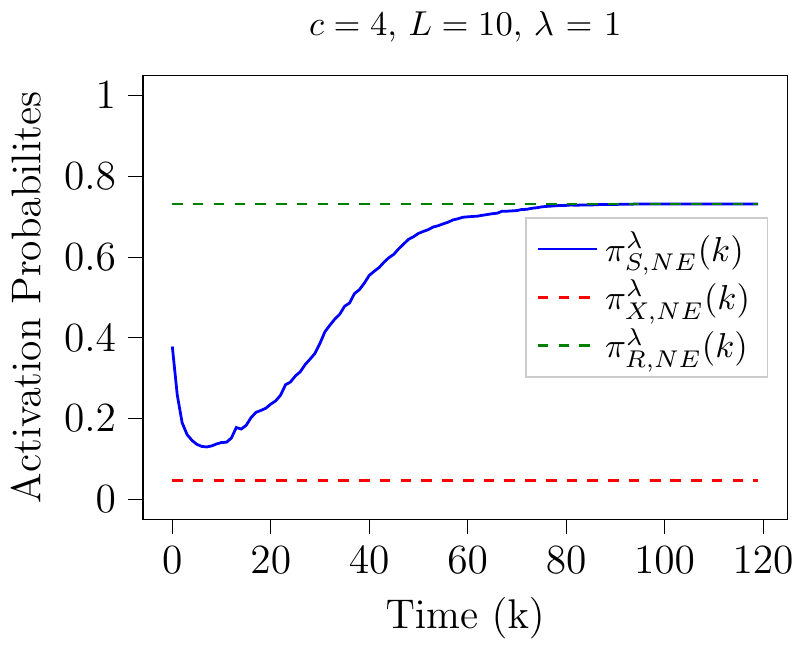}
	\includegraphics[scale=0.5]{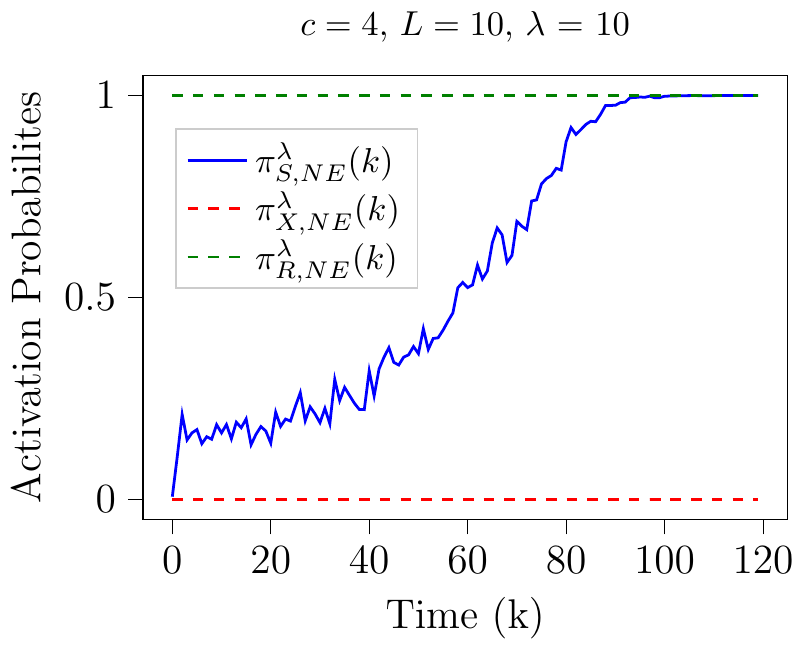}
	\caption{\small Evolution of the proportion of nodes in infected ($\Xt$) (in red) and recovered ($\Rt$) (in green) states, both actual proportions averaged over $75$ independent runs (solid lines) and under the DBMF approximation (dashed lines), are shown. The bottom row shows the evolution of the activation probabilities at the QRE in the A-SIR epidemic model.}
	\label{fig:asir_detailed}
\end{figure*}

We consider a set of $n=100$ nodes and set the rate of infection $\beta = 0.3$, rate of recovery $\delta = 0.15$ and degree $d=5$ for all the nodes. We assume that at $k=0$, 10 nodes are infected and the remaining are susceptible, and simulate the evolution of the epidemic states when the activation probabilities of the nodes are given by the QRE characterization in \eqref{eq:qre_asir_prob}. In Figure \ref{fig:asir_detailed}, we plot the proportion of nodes in infected ($\Xt$) and recovered ($\Rt$) states; \rev{both actual proportions averaged over $75$ independent runs (shown in solid lines) and under the DBMF approximation (shown in dashed lines).} The evolution of the state-dependent activation probabilities for different values of $\lambda$ are also shown in the bottom panel. The proportions under the DBMF approximation closely track the actual proportions of nodes in different epidemic states in all cases. 

From the figures in the left panel of Figure \ref{fig:asir_detailed}, we observe that if $\lambda = 0$, the activation probability is constant at $0.5$ for all nodes irrespective of their epidemic states. This corresponds to the setting where nodes act completely randomly (independent of the utility they may obtain by activating or not activating).  Thus, the state evolution is independent of parameters $c$ and $L$. In this case, there is a sharp initial increase in the proportion of infected nodes and almost all nodes recover by $k=40$. 

As $\lambda$ increases, nodes are increasingly likely to take actions that give a higher utility. Nodes that have recovered receive a higher utility upon activation and thus, the activation probability for recovered nodes is approximately $1$ when $\lambda = 10$. The equilibrium activation probability of infected nodes depends on the parameter $c$ which captures the penalty or cost of activation when a node is infected (or incentives offered to self-isolate). When $c > 1$, the activation probability is approximately $0$ for sufficiently high $\lambda$ . The plots in the bottom row of Figure \ref{fig:asir_detailed} illustrate this for $c = 4$. 

The activation probabilities of susceptible nodes at the QRE depend on the proportion of infected nodes ($\bar{x}$) and is illustrated in Figure \ref{fig:asir_home_actprob} for different values of $\lambda$ and $L$. We observe the following characteristic from Figure \ref{fig:asir_home_actprob}: 
\begin{itemize}
\item as $\bar{x}$ increases, susceptible nodes encounter a larger risk of becoming infected when they activate, and consequently choose a smaller activation probability,
\item when $\lambda$ is large, there is sharp transition in equilibrium activation probability around the value of $\bar{x}$ at which the nodes are indifferent between activating or not activating (i.e., around $\bar{x}$ where $\Delta u_{d,\St}(\pi^\lambda_{\NEt},\bar{x}) = 0$), and
\item this threshold is inversely proportional to the magnitude of loss $L$.
\end{itemize}

The above characteristics impact the evolution of epidemic states as shown in Figure \ref{fig:asir_detailed}. The plots in the top row show the evolution for $c = 4$ and $L = 2$ which corresponds to the plot in the left panel of Figure \ref{fig:asir_home_actprob}. As the proportion of infected nodes exceeds $0.4$, the activation probability of susceptible nodes decreases to $0$, and consequently, the epidemic dies down. When $\lambda = 10$, the activation probability drops sharply which is reflected in the sharp decline in the proportion of nodes in state $\Xt$. 

\begin{figure*}[tb]
	\centering
	\includegraphics[scale=0.6]{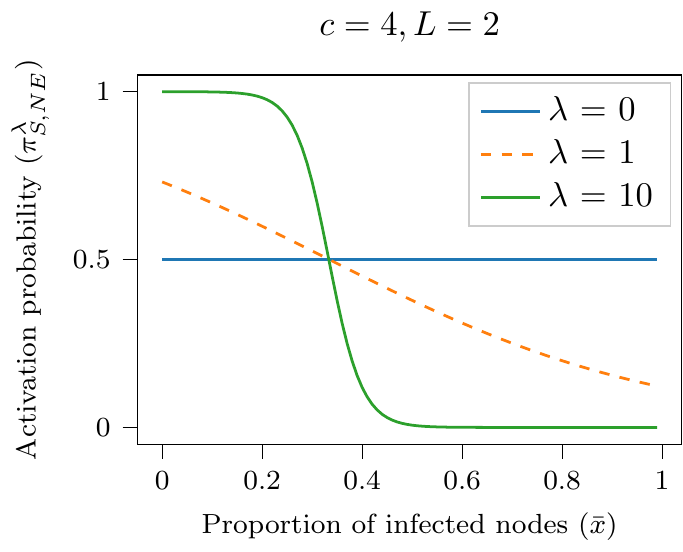}
	\includegraphics[scale=0.6]{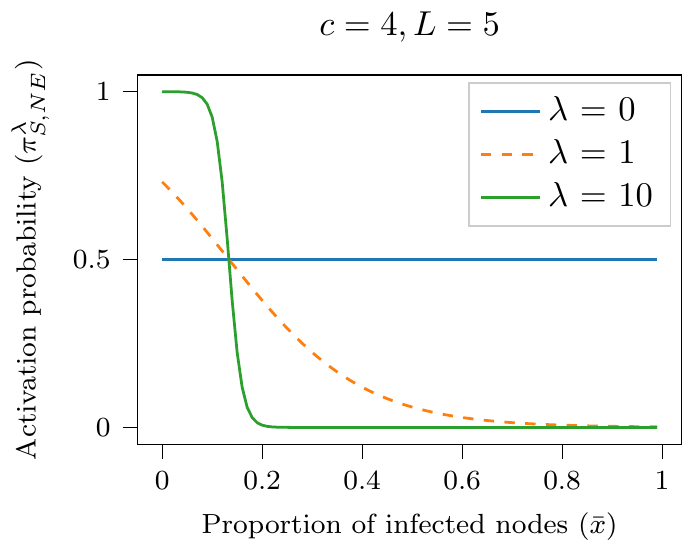}
	\includegraphics[scale=0.6]{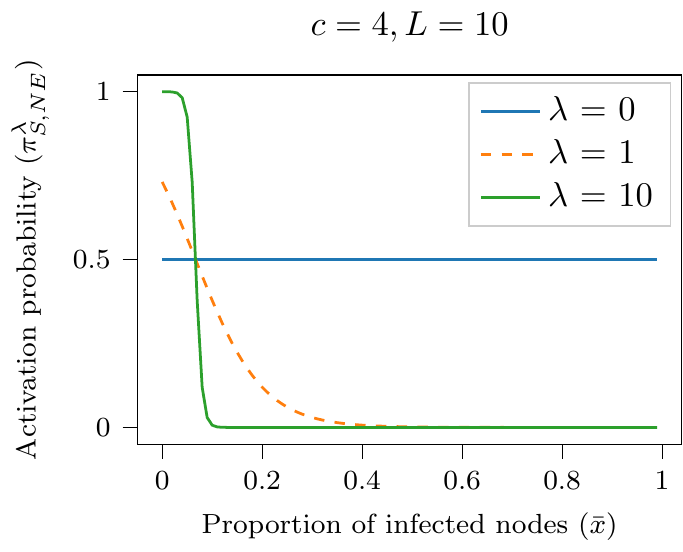}
	\caption{\small Activation probability for susceptible nodes at the QRE as a function of proportion of infected nodes for different values of $\lambda$ and $L$}
	\label{fig:asir_home_actprob}
\end{figure*}

When $c = 4$ and $L = 10$, the evolution of states and the activation probability with time is shown in the middle and bottom rows of Figure \ref{fig:asir_detailed}. Note from the plot on the right panel of Figure \ref{fig:asir_home_actprob} that the transition in activation probability occurs around the threshold $\bar{x} = 0.1$. Specifically, when $\bar{x}(k)$ is slightly larger than the threshold, the susceptible nodes reduce their activation probabilities. However, this reduction results in $\bar{x}(k)$ decreasing below the threshold which then leads to susceptible nodes increasing their activation probabilities. As a result, the proportion of nodes in state $\Xt$ fluctuates around the threshold value for a significant amount of time; any reduction in infected proportion is counteracted by an increase in activation probabilities by self-interested nodes. The plots in the right most panel in the middle and bottom rows in Figure \ref{fig:asir_detailed} clearly illustrate this phenomenon. Thus, activation decisions made by self-interested nodes can lead to a {\it flattening} of the proportion of infected nodes. 

\rev{The above discussion has the following implications for policy makers. 
\begin{itemize}
\item Recall that the parameter $c$ captures the penalty imposed on infected nodes for breaking quarantine or isolation protocols. Our results suggest that this cost should be set higher than the benefit of activation for effective adherence to such protocols. 
\item Our results also suggest that the infected proportion tends to remain in the vicinity of the threshold at which susceptible nodes significantly reduce their activation probabilities. While a larger value of $L$ leads to a smaller threshold, and hence a smaller peak infection level, it could lead to the epidemic sustaining in the population for a longer duration. Therefore, it is critical to set this parameter appropriately to keep the peak infection low depending on the available healthcare facilities. This parameter could also be adapted as healthcare resources are augmented. 
\end{itemize}}

%%%%%%%%%%%%%%%%%%%%%%%%%%%%%%%%%
%%%%%%%%%%%%%%%%%%%%%%%%%%%%%%%%%
%%%%%%%%%%%%%%%%%%%%%%%%%%%%%%%%%

\subsection{A-SAIR epidemic} 

In the previous subsection, we examined the impacts of cost parameters $c$ and $L$ and the logistic choice parameter $\lambda$ on the evolution of epidemic states in the A-SIR epidemic model. We now study the A-SAIR epidemic model and we focus on understanding the impacts of asymptomatic carriers and heterogeneous node degrees. 

\subsubsection{A-SAIR epidemic with homogeneous node degrees} 

\begin{figure*}[tb]
	\centering
	\includegraphics[scale=0.5]{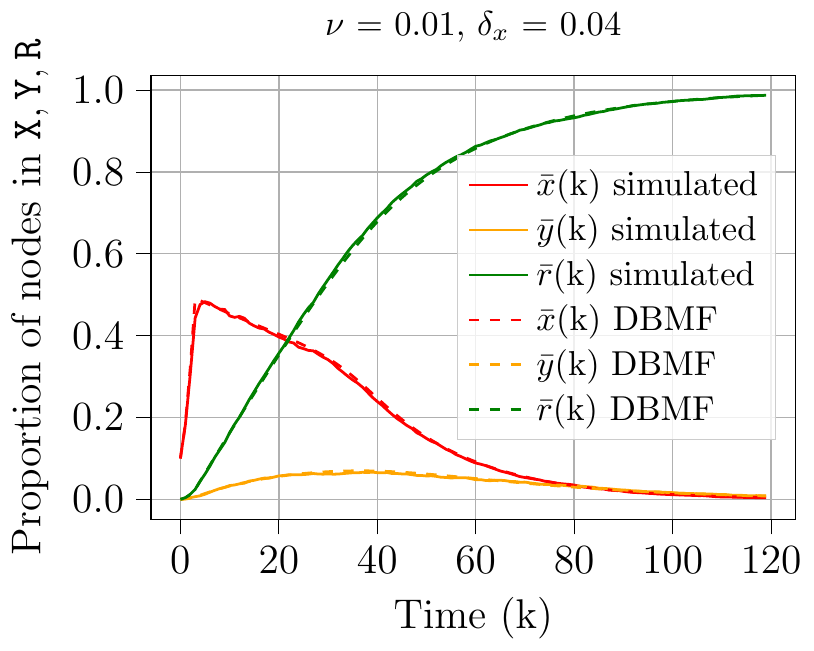}
	\includegraphics[scale=0.5]{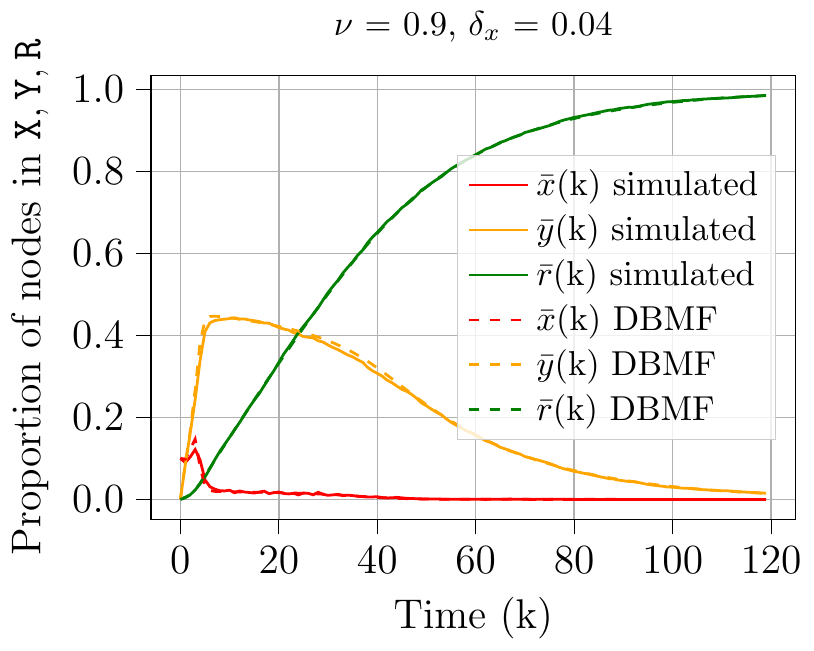}
	\includegraphics[scale=0.5]{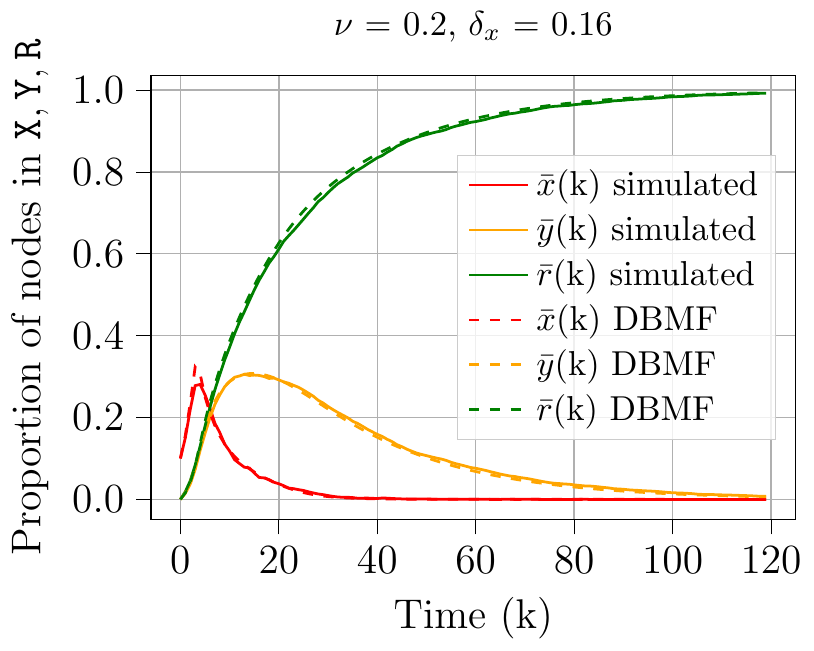} \\[2mm]
	\includegraphics[scale=0.5]{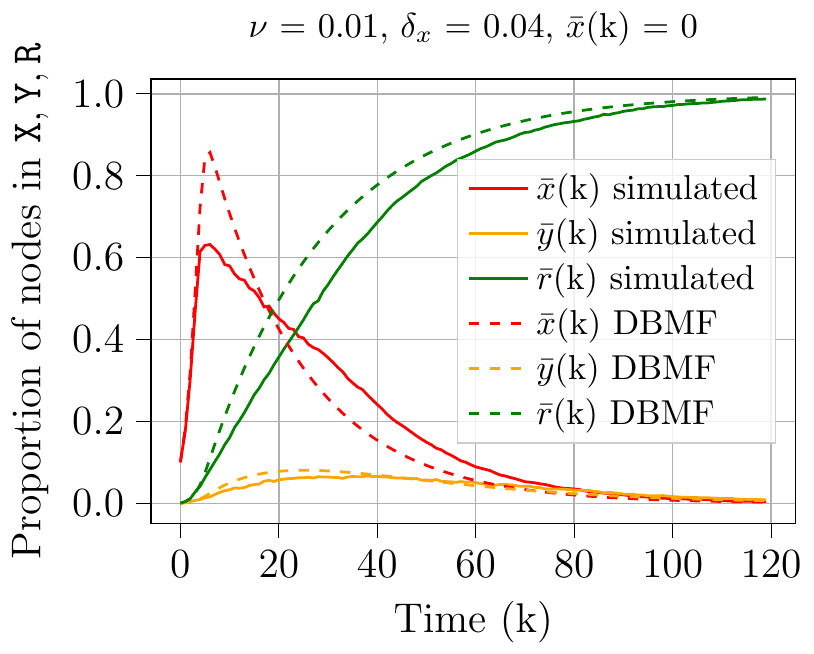}
	\includegraphics[scale=0.5]{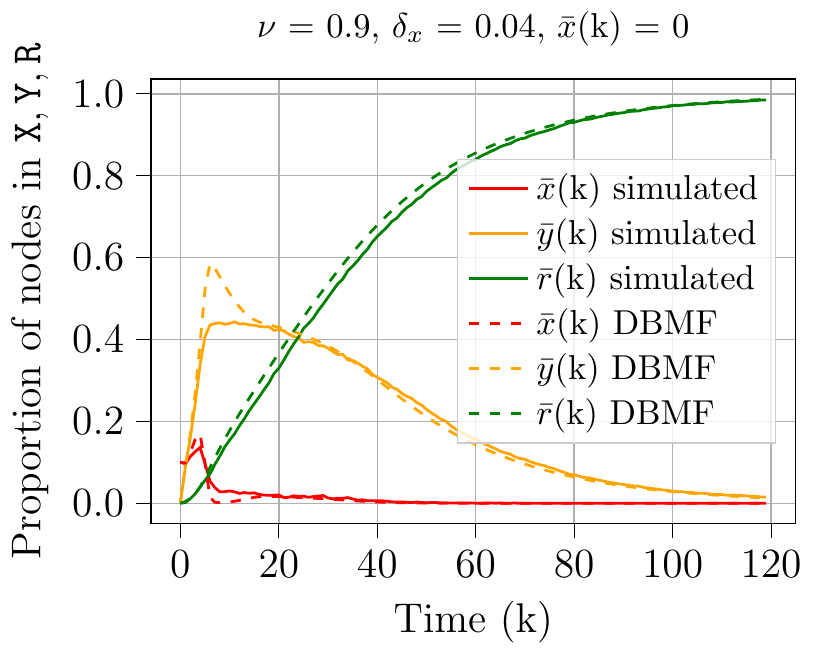}
	\includegraphics[scale=0.5]{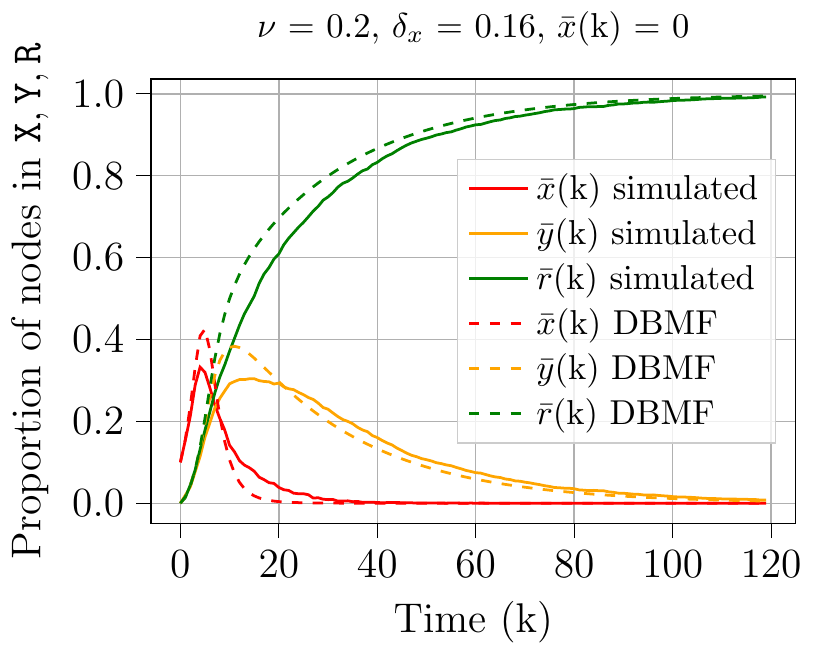} \\[2mm]
	\includegraphics[scale=0.5]{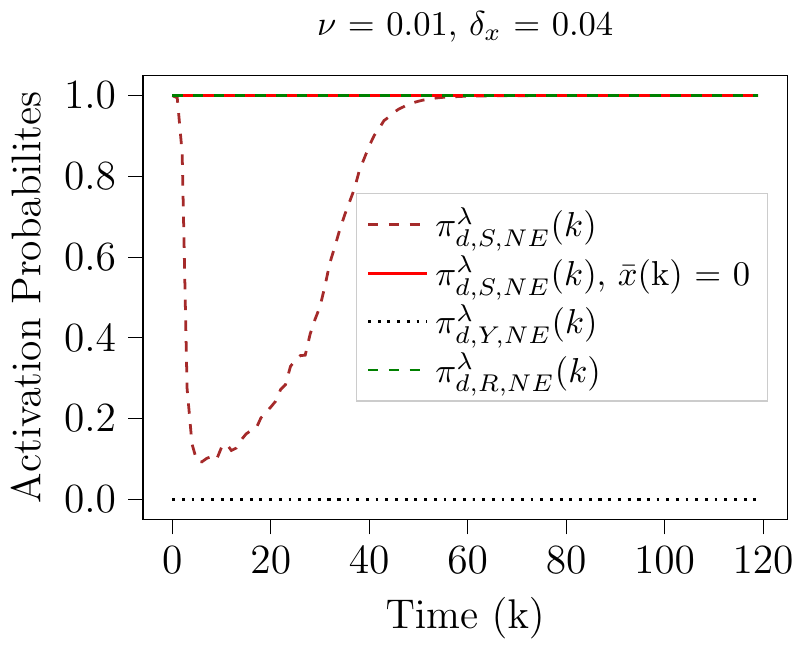}
	\includegraphics[scale=0.5]{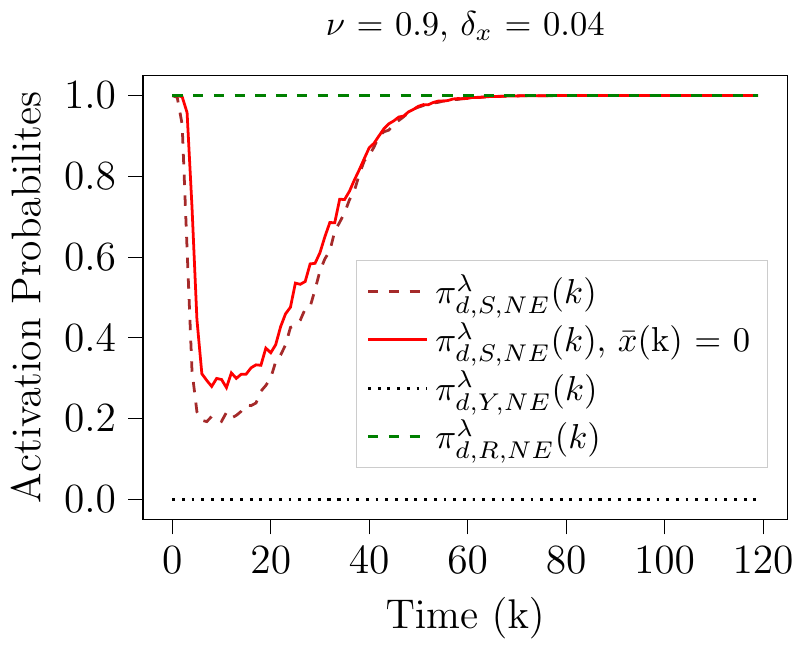}
	\includegraphics[scale=0.5]{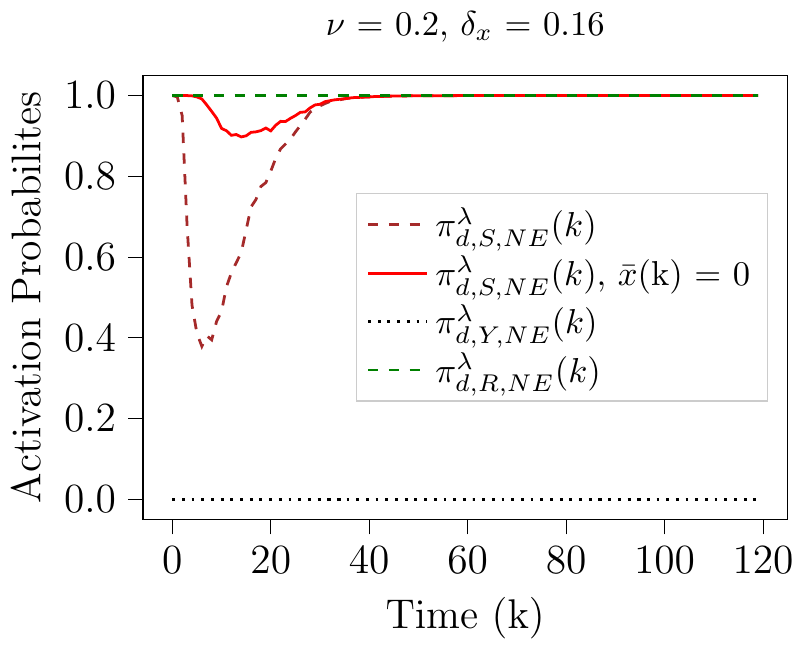}
	\caption{\small Evolution of the proportion of nodes in asymptomatic ($\Xt$), symptomatic ($\Yt$) and recovered ($\Rt$) states (both actual proportions averaged over $50$ independent runs and under the DBMF approximation) and the evolution of the state-dependent activation probabilities at the QRE in the A-SAIR epidemic model. When nodes take activation decisions without being aware of the proportion of asymptomatic nodes, we denote it by $\bar{x}(k) = 0$.}
\label{fig:asiyr_detailed}
\end{figure*}

We consider a set of $n=100$ nodes and set the infection rates $\beta_x = \beta_y = 0.1$, recovery rate $\delta_y = 0.04$ and degree $d = 5$ for all the nodes. \rev{The ratio $\beta / \gamma = 2.5$ corresponds to the basic reproduction number of several infectious diseases such as common cold and SARS \cite{world2003consensus}.} In addition, we choose $\lambda = 10, c = 2, L = 5$. Figure \ref{fig:asiyr_detailed} shows the evolution of proportion of nodes in states $\Xt, \Yt$ and $\Rt$ and the activation probabilities at the QRE (from Proposition \ref{prop:qre_asiyr}) for different values of $\delta_x$ and $\nu$. In order to examine the impacts of asymptomatic carriers on epidemic evolution, we consider a variant where nodes are not aware of the proportion of asymptomatic nodes, i.e., they choose the activation probabilities by setting $\bar{x}(k) = 0$. The evolution of the epidemic states in this case are shown in the middle row of Figure \ref{fig:asiyr_detailed}. The corresponding activation probabilities are shown in the bottom row of Figure \ref{fig:asiyr_detailed}. 

We plot the evolution of states under the DBMF approximation and by taking the average of epidemic states across $50$ independent simulations of the activity-driven epidemic model. The figures show that the DBMF approximation is largely accurate. 

The plots in the left panel of Figure \ref{fig:asiyr_detailed} correspond to when transition rate $\nu = 0.01$. The plot in the left panel of the top row shows that asymptomatic nodes rarely become symptomatic, and largely recover without ever exhibiting symptoms. The plot in the left panel of the middle row shows that if the nodes are not aware of $\bar{x}(k)$, they continue to activate with probability close to $1$, and as a result, there is a much larger peak in the asymptomatic proportion compared to when nodes are aware of $\bar{x}(k)$. 

In contrast, when $\nu = 0.9$, asymptomatic nodes quickly transition to the symptomatic state as shown in the plots in the middle column. In this case, nodes not being aware of $\bar{x}(k)$ has a more benign effect since most of the infected nodes are symptomatic. When the recovery rate for asymptomatic nodes is higher (i.e., $\delta_x = 0.16$), the epidemic dies quickly as expected and observed in the plots on the right column. 

\subsubsection{A-SAIR epidemic with heterogeneous node degrees} 

\begin{figure*}[tb]
	\centering
	\includegraphics[scale=0.5]{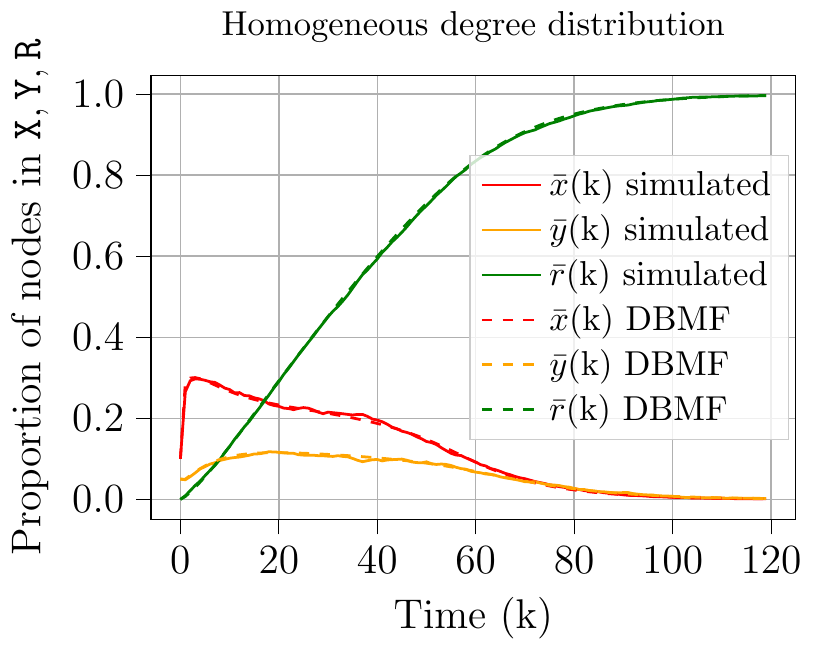}
	\includegraphics[scale=0.5]{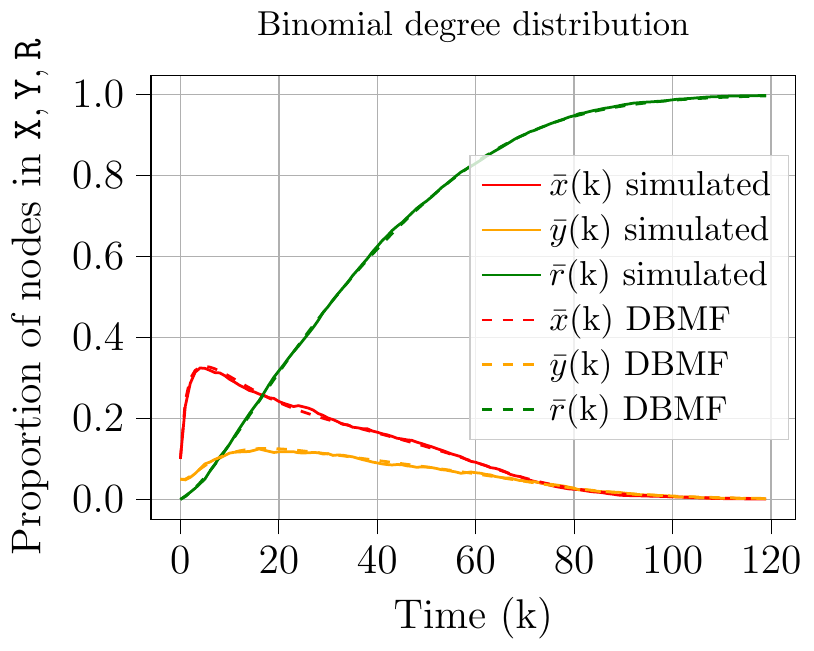}
	\includegraphics[scale=0.5]{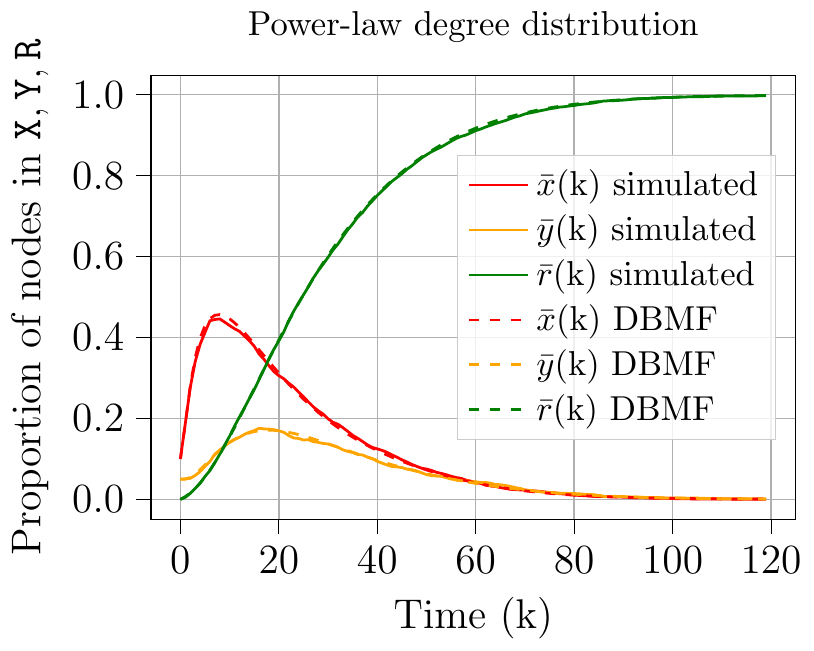} \\[2mm]
	\includegraphics[scale=0.5]{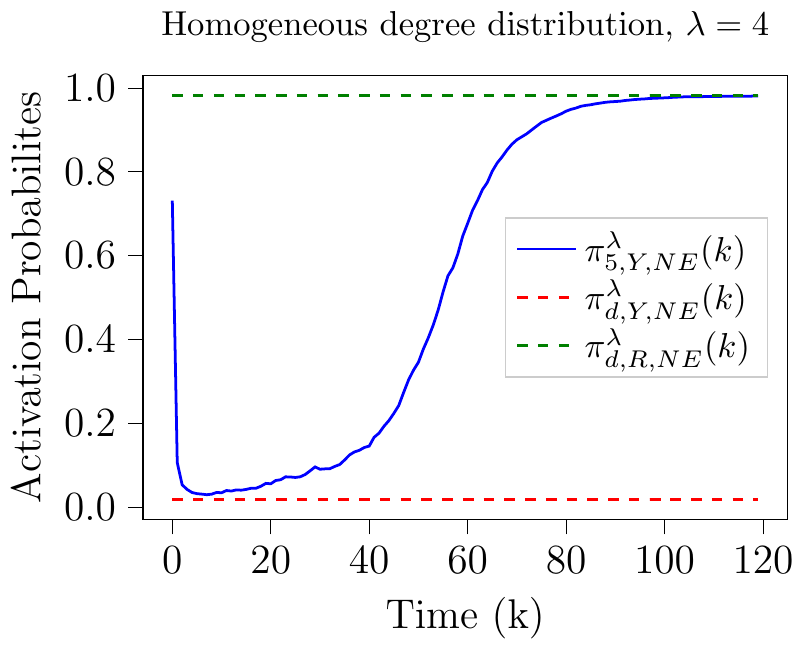}
	\includegraphics[scale=0.5]{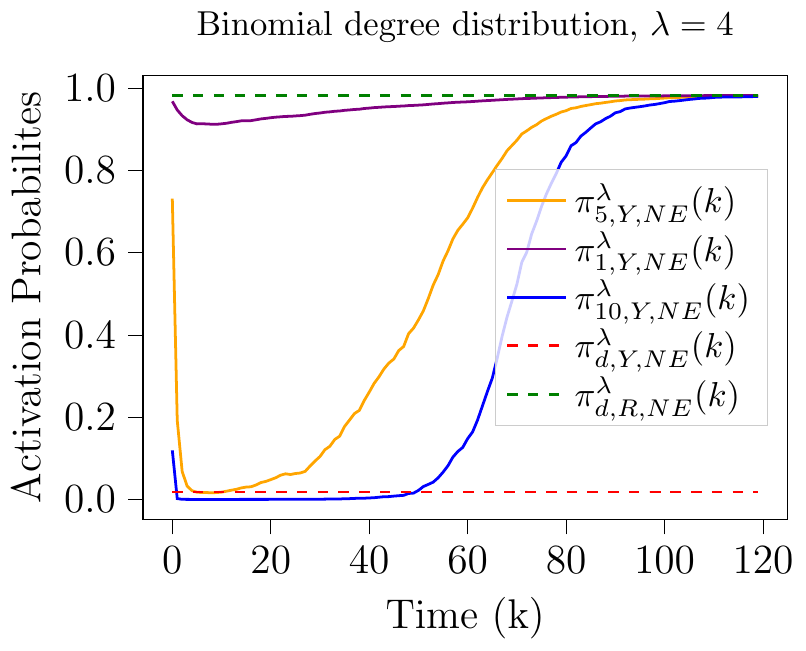}
	\includegraphics[scale=0.5]{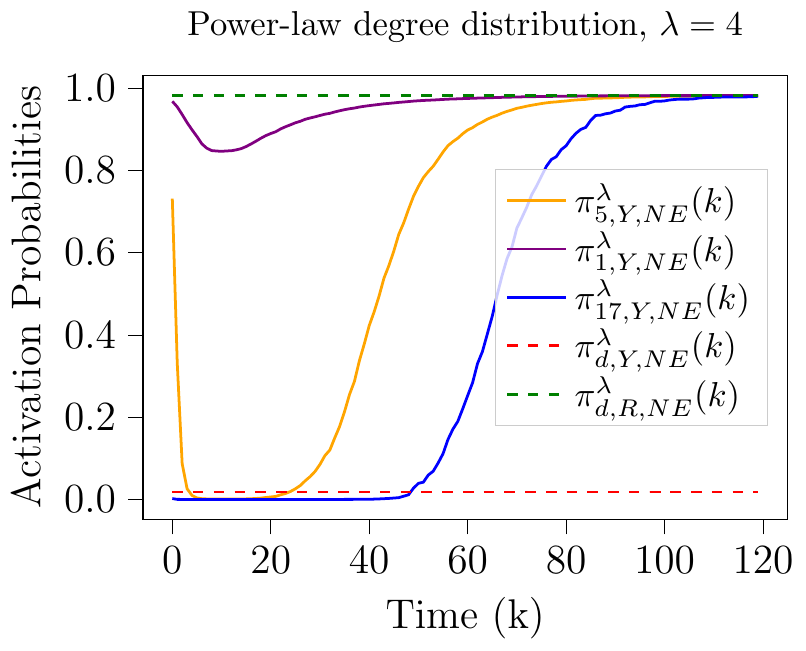}
	\caption{\small Evolution of the proportion of nodes in asymptomatic ($\Xt$), symptomatic ($\Yt$) and recovered ($\Rt$) states (both actual proportions averaged over $50$ independent runs and under the DBMF approximation) and the evolution of the state-dependent activation probabilities at the QRE in the A-SAIR epidemic model.}
	\label{fig:asir_detailed_hetero}
\end{figure*}

We now investigate the evolution of the A-SAIR epidemic in its full generality under game-theoretic activation. We consider a set of $n=100$ nodes and set $\beta_x = \beta_y = 0.25, \delta_x = 0.02, \delta_y = 0.1$ and transition rate $\nu = 0.05$. We set the logistic choice parameter $\lambda = 4$ and the loss parameters $L = 4$ and $c = 2$. We compare the epidemic evolution when all nodes have degree $5$ (the homogeneous case) with settings where degrees have a Binomial distribution and a Power-law distribution with exponent $2$. The above distributions characterize the degree distributions in several important classes of random graphs \cite{newman2010networks}. We choose the parameters of these two distributions with $d_{\max} = 20$ such that mean degree approximately $5$. Figure \ref{fig:asir_detailed_hetero} shows the results obtained with the above parameters.

The evolution of the proportion of nodes in asymptomatic ($\Xt$), symptomatic ($\Yt$) and recovered ($\Rt$) states (both actual proportions averaged over $50$ independent runs and under the DBMF approximation) and the evolution of the state-dependent activation probabilities at the QRE with time are shown in the top and bottom panels of Figure \ref{fig:asir_detailed_hetero}. Under the chosen parameters, recovered nodes activate with probability approximately $1$ and symptomatic nodes activate with probability close to $0$. 

When all nodes have homogeneous degrees (figures in the left panel), an initial increase in the infected proportion causes the activation probability of susceptible nodes to fall close to $0$ and eventually the epidemic declines. Under the Binomial distribution (middle panel), the activation probabilities depend on the degrees of the nodes as well. Specifically, nodes with small degrees continue to activate at a very high rate which leads to a larger peak in the infected proportion compared to the homogeneous case. Under the Power-law degree distribution (right panel), a significant proportion of the nodes have very small degree, and as a result, a large proportion of nodes continue to activate at a high rate leading to an even sharper peak in the infected proportion.

%%%%%%%%%%%%%%%%%%%

\subsection{A-SIS Epidemic}\label{sec:sim_asis}

We now examine the evolution of the A-SIS epidemic under game-theoretic activation. We consider $n = 100$ nodes and set the infection rate $\beta = 0.2$, recovery rate $\delta = 0.4$, degree $d = 4$, penalty parameter $c = 10$ and $\lambda = 20$. The evolution of the proportion of infected nodes, both in actual simulations and under the DBMF approximation, for different values of $L$ is shown in Figure \ref{fig:asis_sim}. 

The plots illustrate the theoretical findings in Proposition \ref{prop:sis_cl_homo}. Specifically, for the considered parameters, there exists an endemic state of the epidemic with $x^*_d = 1 - \frac{\delta}{\beta d} = 0.5$. When the loss parameter $L = 2$, the threshold $\frac{1}{L\beta d} = 0.625 > x^*_d$. Accordingly, the infected proportion under the DBMF approximation settles around $x^*_d$ and infected proportion under actual simulation settles at a slightly smaller value. When $L = 5$, we have $\frac{1}{L\beta d} = 0.25 < x^*_d$. We observe that the actual infected fraction and the DBMF approximation oscillates around $\frac{1}{L\beta d}$ as discussed in Section \ref{section:sis_closedloop}. Any increase in $x(k)$ beyond $\frac{1}{L\beta d}$ results in activation probability dropping to $0$ which causes $x(k)$ to drop down to $\frac{1}{L\beta d}$. Finally, when $L = 10$, the threshold $\frac{1}{L\beta d} = 0.125$. The infected proportion decreases from the initial value of $0.2$ to the threshold value. 

Thus, for sufficiently large $L$ and rationality parameter $\lambda$, susceptible individuals reduce their activation probability close to $0$ when the infection exceeds beyond the threshold $\frac{1}{L \beta d}$. Due to this protective action, the infected proportion under game-theoretic activation is considerably smaller than $x^*_d$. 

\begin{figure*}[tb]
	\centering
	\includegraphics[scale=0.5]{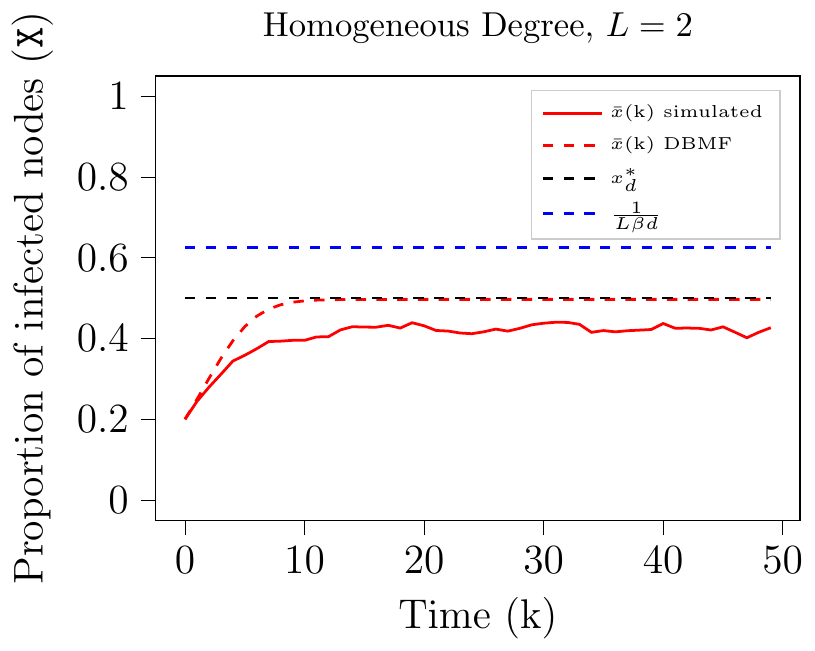}
	\includegraphics[scale=0.5]{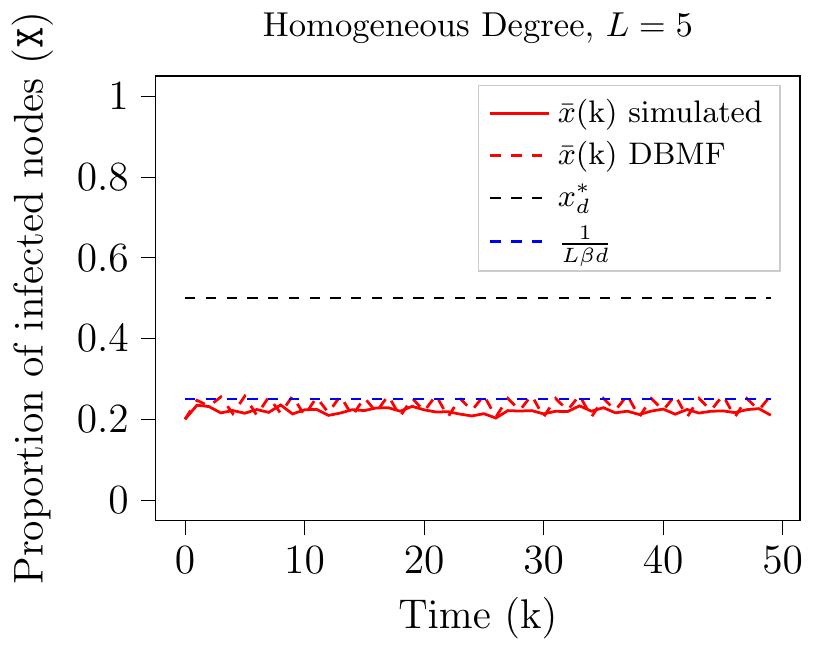}
	\includegraphics[scale=0.5]{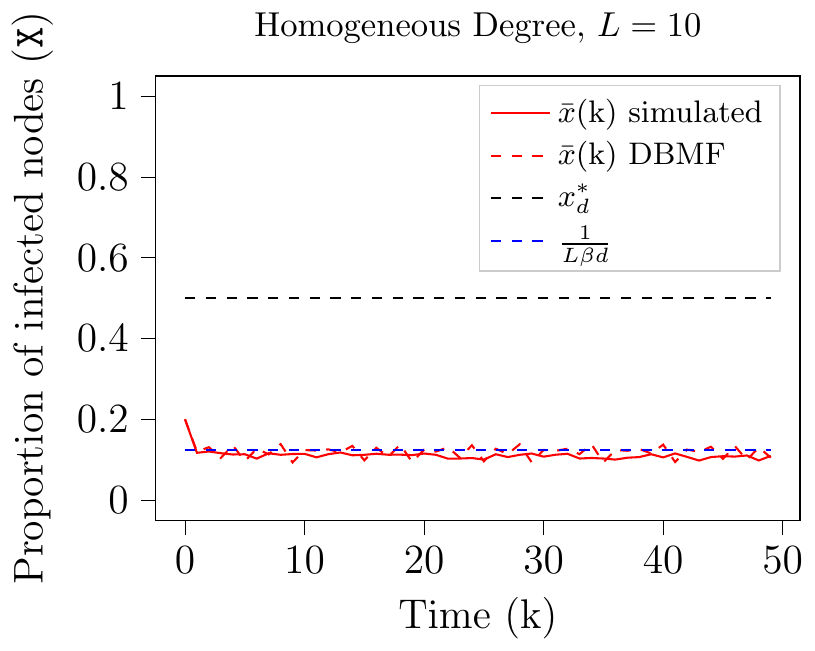} 
	\caption{\small Evolution of the proportion of infected nodes (both actual proportions averaged over $50$ independent runs and under the DBMF approximation) in the A-SIS epidemic model.}
	\label{fig:asis_sim}
\end{figure*}

\section{Discussion and Conclusion}\label{sec:discussion} Our work is one of the first to conduct a rigorous and comprehensive analysis of the impacts of game-theoretic activation on epidemic spread over dynamical networks. One of the main insights from our results is that, in both A-SAIR and A-SIS epidemic models, the infection can persist in the population for a long time under game-theoretic activation as any increase or decrease in the infected proportion is counteracted by a decrease or increase in the activation probabilities, respectively. There are several promising directions for future research. 
\begin{itemize}
\item \rev{Analyzing equilibria when players are not myopic, but rather maximize payoffs over a future time horizon, is an important problem. In an ongoing work, we build upon this paper and investigate the behavior of agents who consider discounted payoffs over a longer time horizon \cite{elokdadynamic}. This topic remains relatively unexplored and is a promising avenue for future research.}
\item The analysis of closed-loop dynamics under game-theoretic activation for the A-SAIR epidemic remains a challenging open problem. 
\item While we show that the equilibrium activation probabilities depend on the proportion of infected nodes, this information may not be available in practice as testing data is often not accurate and represents the epidemic state with a certain delay. Therefore, the impacts of inaccurate and delayed information on activation probabilities and epidemic evolution should be investigated. 
\item In this work, we have assumed that edges or connections are formed with probability $1$. This assumption can potentially be relaxed by letting nodes choose whether to accept incoming connections or not in a strategic manner.
\item Much of the earlier work on inferring epidemic states and resource allocation for epidemic containment rely on centralized algorithms (e.g., modifying infection and recovery rates) and do not account for behavioral changes in the activation pattern of the nodes \cite{nowzari2016analysis,sagar2020_2,hota2020closed}. The analysis carried out in our work should be integrated with inference, estimation and optimal resource allocation schemes to obtain more effective countermeasures.  
\item Finally, we should explore how to infer different parameters of interest (such as $\lambda, L, c$) from real data, such as mobility data made available for various activities in \cite{Mobility_Google}, and validate the theoretical predictions.
\end{itemize}
We hope our work stimulates further investigations along the above avenues.

\section*{Acknowledgments}
We thank Prof. Shreyas Sundaram (Purdue University), Prof. Philip Par{\'e} (Purdue University), Dr. Bala Kameshwar Poolla (NREL), Prof. Ceyhun Eksin (TAMU) and Prof. D. Manjunath (IIT Bombay) as well as anonymous reviewers for helpful suggestions. 

\bibliographystyle{IEEEtran}
\bibliography{refs,refs_new}

\myclearpage
\appendix
\section{Preliminary results and proofs pertaining to the A-SIS epidemic model}

\rev{We here introduce a few preliminary results regarding the equilibria and stability of the SIS epidemic model followed by the detailed proofs of some of the results from Section \ref{section:a-sis}.}

\subsection{Equilibria and stability of discrete-time networked SIS epidemic}
\label{sec:appendix_prior}

\rev{Consider the discrete-time SIS epidemic dynamic on a network with $N$ nodes. The proportion of infected individuals at node $i \in [N]$\footnote{\rev{We define $[N] := \{1,2,\ldots,N\}$.}} at time $k$, denoted $x_i(k)$, evolves as
\begin{equation}\label{eq:dt_sis_prior}
x_i(k+1) = x_i(k) + h \big[ (1-x_i(k)) \sum^N_{j = 1} \beta_{ij} x_j(k) - \delta_i x_i(k) \big], \qquad i \in [N],
\end{equation}
where $\beta_{ij} \geq 0$ denotes the infection rate from $j$ to $i$, $\delta_i$ denotes the recovery rate at node $i$ and $h > 0$ is the sampling parameter. We now state a few assumptions followed by a characterization of the equilibria of the above dynamics which depends critically on the spectral radius of the matrix $\Ib - hD + hB$, where $\Ib$ is the identity matrix of dimension $N$, $B$ denotes the square matrix with entry $(i,j)$ being $\beta_{ij}$, and $D$ denotes the diagonal matrix with entries being the $\delta_i$ parameters.} 

\begin{assumption}\label{ass:dt_sis_prior}
\rev{The following are true. 
\begin{enumerate}
\item The initial infected proportion $x(0) \in [0,1]^N$.  
\item The parameters satisfy $\beta_{ij} \geq 0, \delta_i \geq 0$. 
\item For every node $i$, $h\delta_i \leq 1$, and $h\sum^N_{j=1} \beta_{ij} \leq 1$.
\item The matrix $B$ is irreducible. 
\end{enumerate}}
\end{assumption}

\begin{theorem}[\cite{pare2018analysis,liu2020stability}]\label{theorem:dt_sis_prior}
\rev{Consider the discrete-time SIS epidemic dynamic given by \eqref{eq:dt_sis_prior} under Assumption \ref{ass:dt_sis_prior}. Then, $[0,1]^N$ is an invariant set for \eqref{eq:dt_sis_prior}. In addition,  
\begin{enumerate}
\item if $\rho(\Ib-hD+hB) \leq 1$, then the disease-free state, $x = \mathbf{0}$, is the only equilibrium of \eqref{eq:dt_sis_prior} and is asymptotically stable with domain of attraction $[0,1]^N$, and
\item if $\rho(\Ib-hD+hB) > 1$, then \eqref{eq:dt_sis_prior} possesses two equilibria, $\mathbf{0}$ and $x^*$ with $x_i^* > 0, \forall i \in [N]$. Furthermore, if $h(\delta_i +\sum^N_{j=1} \beta_{ij}) \leq 1, \forall i \in [N]$, then the endemic equilibrium $x^*$ is asymptotically stable with domain of attraction $[0,1]^N \setminus \{\mathbf{0}\}$. \end{enumerate}}
\end{theorem}

\rev{The existence of the endemic equilibrium $x^*$ was shown in \cite{pare2018analysis} under Assumption \ref{ass:dt_sis_prior}. The asymptotic stability of this equilibrium was established recently in \cite{liu2020stability} under the additional assumption $h(\delta_i +\sum^N_{j=1} \beta_{ij}) \leq 1, \forall i \in [N]$. We leverage the above theorem to establish several results for our setting.}

\subsection{Omitted proofs from Section \ref{section:asis-mfa}}
\label{sec:appendix_proofs}

\subsubsection{Proof of Theorem \ref{thm:asis}}
\begin{proof}
We compute expectation on both sides of \eqref{eq:asis_markov} and obtain
\begin{align}
x_i(k+1) & = (1-\delta) x_i(k) + \Eb\Big[(1-X_i(k)) \big[1 - \underset{j \neq i}{\Pi} \left[1 - A_{ij}(k) X_j(k)N_{\beta} \right]\big]\Big] \nonumber
\\ & \leq (1-\delta) x_i(k) + \sum_{j \neq i} \beta \Eb\big[ A_{ij}(k) X_j(k) (1-X_i(k)) \big]
\\ & = (1-\delta) x_i(k) + \sum_{j \neq i} \beta \Pb(A_{ij}(k)=1|\SXt^k_{ij})\Pb(\SXt^k_{ij}), \label{eq:asis_mfa2_x}  
\end{align}
following Weierstrass product inequality. Here, $\SXt^k_{ij}$ denotes the event that node $v_i(k) \in \St$ (i.e., $X_i(k) = 0$) and node $v_j(k) \in \Xt$ (i.e., $X_j(k) = 1$) as in \eqref{eq:def_SXtij}. We now compute
\begin{align}
\Pb(A_{ij}(k)=1|\SXt^k_{ij}) & = \Pb(\Gamma^k_{i \to j} \cup \Gamma^k_{j \to i}| \SXt^k_{ij}) \nonumber
\\ & = 1-[1-\Pb(\Gamma^k_{i \to j} | \SXt^k_{ij})][1-\Pb(\Gamma^k_{j \to i} | \SXt^k_{ij})] \nonumber
\\ & = 1-[1-\pi_{\St,i}(k)\bar{d}_i][1-\pi_{\Xt,j}(k)\bar{d}_j], \label{eq:mfa_asis_int1}
\end{align}
where $\Gamma^k_{i \to j}$ is as defined in \eqref{eq:def_Gammatij}. Note that the probability of an activated node $v_i$ choosing a specific node $v_j$ is $\bar{d}_i = \frac{d_i}{n-1}$. 

We now approximate $\Pb(\SXt^k_{ij}) \simeq (1-x_i(k)) x_j(k)$ assuming that states of nodes $v_i$ and $v_j$ are independent. The result now follows upon substituting \eqref{eq:mfa_asis_int1} in \eqref{eq:asis_mfa2_x}. 
\end{proof}

\subsubsection{Proof of Proposition \ref{prop:sis_endemic_rho}}
\begin{proof}
\rev{Note that when the activation probabilities are given by static constants $\pi_{d,\St}$ and $\pi_{t,\Xt}$, the dynamics in \eqref{eq:dbmf_asis} is analogous to the dynamics in \eqref{eq:dt_sis_prior}. In particular, drawing parallel with the notation in \eqref{eq:dt_sis_prior}, we have $h=1$, $N = |\DD|$, and $\beta_{ij}$ being analogous to $nm_t \beta [1-(1-\pi_{d,\St} \bar{d})(1-\pi_{t,\Xt}\bar{t})]$ between degrees $d$ and $t$. Since $\beta, \delta \in (0,1]$ and $\pi_{d,\St} > 0, \pi_{d,\Xt} > 0$, it is easy to see that the requirements in Assumption \ref{ass:dt_sis_prior} are satisfied with the exception of the third point. We now verify this remaining point. For a given degree $d \in \DD$, we compute
\begin{align*}
\sum_{t \in \DD} nm_t \beta \big[1-(1-\pi_{d,\St} \bar{d})(1 - \pi_{t,\Xt} \bar{t})\big] & \leq \sum_{t \in \DD} nm_t \beta \big[\pi_{d,\St} \bar{d} + \pi_{t,\Xt} \bar{t}\big]
\\ & \leq \sum_{t \in \DD} nm_t \beta (\bar{d} + \bar{t}) 
\\ & \simeq \sum_{t \in \DD} (m_t \beta d + m_t \beta t) = \beta (d+d_{\avg}),
\end{align*}
where we assume $n \simeq n-1$ (which is accurate when the number of nodes is sufficiently large) and $\pi_{d,\St} \leq 1, \pi_{t,\Xt} \leq 1$ being probabilities. Thus, when $\beta(d_{\avg} + d_{\max}) \leq 1$, Assumption \ref{ass:dt_sis_prior} is satisfied. The result now follows from Theorem \ref{theorem:dt_sis_prior}.} 
\end{proof}

%%%%%%%%%%%%%%%%%%%%%%%%%%%%%%%%%
%%%%%%%%%%%%%%%%%%%%%%%%%%%%%%%%%
%%%%%%%%%%%%%%%%%%%%%%%%%%%%%%%%%

\subsection{Omitted proofs from Section \ref{section:sis_closedloop}} 
\label{sec:appendix_sishet_proofs}

\subsubsection{Proof of Proposition \ref{prop:sis_cl_homo}}

\begin{proof}
\rev{When all nodes have degree $d$, i.e., $m_d = 1$, then \eqref{eq:cl_sis_rational} yields 
\begin{align}\label{eq:cl_sis_rational_homo}
x_d(k+1) = \begin{cases}
(1-\delta) x_d(k), \qquad & \text{if } x_d(k) \in (\frac{1}{L\beta d},1],
\\ (1-\delta) x_d(k) + (1-x_d(k)) \beta d x_d(k), & \text{if } x_d(k) \in [0, \frac{1}{L\beta d}],
\end{cases} 
\end{align}
as $\bar{x}(k) = x_d(k)$. Following \cite{pare2018analysis}, it is easy to see that $[0,1]$ is invariant for the above dynamic when $\delta, \beta \in (0,1]$ and $\beta d \leq 1$. Furthermore, when $x_d(k) > \frac{1}{L\beta d}$, we have $x_d(k+1) = (1-\delta)x_d(k) < x_d(k)$, i.e., the infected proportion is strictly decreasing. We now examine the following cases.}

\rev{\noindent{\bf Case 1: $\delta \geq \beta d$.} In this case, for $x_d(k) \leq \frac{1}{L\beta d}$, we have
$$ x_d(k+1) - x_d(k) \leq (-\delta+(1-x_d(k))\beta d) x_d(k) \leq (\beta d- \delta) x_d(k).$$
Note that for $x_d(k) > 0$, we have $-\delta + (1-x_d(k))\beta d < 0$. As a result, $x_d(k+1) < x_d(k)$ for all $x_d(k) \in (0,1]$. Therefore $0$ is the only equilibrium point with domain of attraction $[0,1]$.}  

\rev{\noindent{\bf Case 2: $\delta < \beta d$ and $x^*_d \leq \frac{1}{L\beta d}$.} When $\delta < \beta d$, then the dynamic
\begin{equation}\label{eq:prop46_int}
x_d(k+1) = (1-\delta) x_d(k) + (1-x_d(k)) \beta d x_d(k) 
\end{equation}
has two equilibrium points: $0$ and $x^*_d = 1 - \frac{\delta}{\beta d}$, and $x^*_d$ has domain attraction $(0,1]$ following \cite{liu2020stability} (see Theorem \ref{theorem:dt_sis_prior}). From \eqref{eq:prop46_int}, we obtain
\begin{align*}
x_d(k+1) - x^*_d & = (1-\delta) x_d(k) + (1-x_d(k)) \beta d x_d(k) - x^*_d
\\ & = (x_d(k)-x^*_d) + ((1-x_d(k)) \beta d - \delta) x_d(k)
\\ & = (x_d(k)-x^*_d) + ((1-x_d(k)) \beta d - (1-x^*_d)\beta d) x_d(k)
\\ & = (x_d(k)-x^*_d) (1-\beta d x_d(k))
\\ \implies |x_d(k+1) - x^*_d| & \leq |x_d(k)-x^*_d|, 
\end{align*}
where we have used the fact that $\delta = (1-x^*_d)\beta d$ and that $\beta d x_d(k) \in (0,1)$. In other words, if $x_d(k) > x^*_d$, then $x_d(k)$ converges to $x^*_d$ monotonically from above, and vice versa. We now consider the dynamic in \eqref{eq:cl_sis_rational_homo} in this parameter regime.}  

\rev{If $x_d(0) \in (0,\frac{1}{L\beta d}]$, then by the above arguments, $|x_d(k+1)-x^*_d|$ decreases monotonically. Therefore, $x_d(k) \in (0, \frac{1}{L\beta d}]$ for all $k > 0$ where it continues to evolve according to \eqref{eq:prop46_int}. As a result, $x_d(k)$ converges to $x^*_d$.}

\rev{If $x_d(0) \in (\frac{1}{L\beta d},1]$, then we have $x_d(k+1) = (1-\delta)x_d(k)$. Therefore, $x_d(k)$ is strictly decreasing for $x_d(k) > \frac{1}{L\beta d}$, and there exists a finite $k^*$ such that $x_d(k^*) \leq \frac{1}{L\beta d}$. For $k > k^*$, $x_d(k)$ follows the dynamic in \eqref{eq:prop46_int}, and from the above arguments, converges to $x^*_d$ without exceeding $\frac{1}{L\beta d}$.}

\rev{\noindent{\bf Case 3: $\delta < \beta d$ and $x^*_d > \frac{1}{L\beta d}$.} If $x_d(0) > \frac{1}{L\beta d}$, then $x_d(k)$ is monotonically decreasing as $x_d(k+1) = (1-\delta)x_d(k)$ in this regime. If $0 < x_d(0) \leq \frac{1}{L\beta d}$, $|x_d(k)-x^*_d|$ is monotonically decreasing, which implies $x_d(k)$ is monotonically increasing as $x^*_d > \frac{1}{L\beta d}$. Consequently, $\frac{1}{L\beta d}$ acts as a sliding surface for $x_d(k)$ under the dynamic \eqref{eq:cl_sis_rational_homo} in this regime. If $x_d(0) = 0$, then $x_d(k) = 0$ for all $k$ as before.}
\end{proof}

\subsubsection{Proof of Proposition \ref{cor:dt_allact}}

\begin{proof}
\rev{The dynamics in \eqref{eq:allact_sis_dyn} is analogous to \eqref{eq:dt_sis_prior} with $h = 1$, $\delta_i = \delta$ and the matrix $B \in \Rb^{|\DD'| \times |\DD'|}$ with entries $[B]_{d,t} = \beta d m_t$. Consequently, the matrix $\Ib - hD + hB$ takes the form $(1-\delta)\Ib + \beta \mathbf{d}\mathbf{m}^\top$, where $\mathbf{d}$ is the vector that contains the degrees of the nodes contained in set $\DD'$ and $\mathbf{m}$ is the vector that contains the mass of nodes with each degree $d \in \DD'$. The matrix $(1-\delta)\Ib + \beta \mathbf{d} \mathbf{m}^\top$ is a rank$-1$ perturbation to an identity matrix, and its spectral radius is $1-\delta + \beta \mathbf{m}^\top \mathbf{d}$. The proof now follows as a consequence of Theorem \ref{theorem:dt_sis_prior}.}
\end{proof}
 
\subsubsection{Proof of Proposition \ref{prop:endemic_monotone}}

\begin{proof}
\rev{If the dynamics with degrees in $\DD'$ admits an endemic state, then $\delta < \beta \sum_{t \in \DD'} (m_t t) \leq \beta \sum_{t \in \DD} (m_t t)$ since $\DD' \subseteq \DD$. Therefore, following Proposition \ref{cor:dt_allact}, the dynamics with degrees in $\DD$ also admits an endemic state.}

\rev{At the endemic equilibrium of \eqref{eq:allact_sis_dyn} with degrees in $\DD$, we have
\begin{align*}
& \delta x^*_d(\DD) = (1-x^*_d(\DD)) \beta d \bar{x}^*(\DD)
\\ \implies & x^*_d(\DD) (\delta + \beta d \bar{x}^*(\DD)) = \beta d \bar{x}^*(\DD)
\\ \implies & x^*_d(\DD) = \frac{\beta d \bar{x}^*(\DD)}{\delta + \beta d \bar{x}^*(\DD)}
\\ \implies & \sum_{d \in \DD} m_d x^*_d(\DD) = \bar{x}^*(\DD) = \sum_{d \in \DD} m_d \frac{\beta d \bar{x}^*(\DD)}{\delta + \beta d \bar{x}^*(\DD)}
\\ \implies & \sum_{d \in \DD} \frac{m_d \beta d}{\delta + \beta d \bar{x}^*(\DD)} = 1,
\end{align*}
since $\bar{x}^*(\DD) > 0$ at the endemic equilibrium. Similarly, at the endemic equilibrium with degrees in $\DD'$, we have 
\begin{align*} 
& 1 = \sum_{d \in \DD'} \frac{m_d \beta d }{\delta + \beta d \bar{x}^*(\DD')}
\\ \implies & \sum_{d \in \DD'} \frac{m_d \beta d }{\delta + \beta d \bar{x}^*(\DD')} = \sum_{d \in \DD} \frac{m_d \beta d }{\delta + \beta d \bar{x}^*(\DD)} \geq \sum_{d \in \DD'} \frac{m_d \beta d }{\delta + \beta d \bar{x}^*(\DD)}
\\ \implies & \bar{x}^*(\DD') \leq \bar{x}^*(\DD).
\end{align*}
This concludes the proof.}
\end{proof}

%%%%%%%%%%%%%%%%%
%%%%%%%%%%%%%%%%%
%%%%%%%%%%%%%%%%%

\subsubsection{Proof of Theorem \ref{thm:dt_cl_sis_main}}

\begin{proof}
\rev{From \eqref{eq:cl_sis_rational}, we obtain
$$ x_d(k+1) \leq (1\!-\delta) x_d(k) \!+ (1-x_d(k)) \beta d \sum_{t \in \DD} m_t x_t(k), $$
i.e., the dynamics \eqref{eq:allact_sis_dyn} with $\DD' = \DD$ acts as an upper bound on \eqref{eq:cl_sis_rational} for all $\bar{x}(k) \in [0,1]$. When $\beta, \delta \in (0,1]$ and $\beta d_p \leq 1$, it is easy to see that Assumption \ref{ass:dt_sis_prior} is satisfied, and $[0,1]^p$ is invariant for the above dynamics, and consequently for \eqref{eq:cl_sis_rational}. When $\delta \geq \beta \sum_{d \in \DD} d m_d$, it follows from Proposition \ref{cor:dt_allact} that $x = \mathbf{0}$ is the unique equilibrium of \eqref{eq:allact_sis_dyn} with region of attraction $[0,1]^p$. The first part of the theorem now follows.} 

\rev{We now consider the second part of the theorem. Recall that $d_1 < d_2 < \ldots < d_p$. Consequently, the thresholds satisfy $\frac{1}{L\beta d_p} < \frac{1}{L\beta d_{p-1}} < \ldots < \frac{1}{L\beta d_1}$. Since $\delta < \beta \sum_{d \in \DD} dm_d$, we have $\bar{x}^{**}(\DD_p) > 0$. Following Proposition \ref{prop:endemic_monotone}, we also have $\bar{x}^{**}(\DD_1) \leq \bar{x}^{**}(\DD_2) \leq \ldots \leq \bar{x}^{**}(\DD_p)$. We distinguish between the following cases.}

\medskip

\noindent{\emph{Case 1: $d_r = d_p$.}} \rev{In this case, we have $\bar{x}^{**}(\DD_p) \in (0,\frac{1}{L\beta d_p}]$. Note that for any $\bar{x}(k) \in (0,\frac{1}{L\beta d_p}]$, all susceptible nodes activate with probability $1$, and the closed-loop dynamics \eqref{eq:cl_sis_rational} is given by
$$ x_d(k+1) = (1\!-\delta) x_d(k) \!+ (1-x_d(k)) \beta d \sum_{t \in \DD} m_t x_t(k), $$
i.e., it coincides with \eqref{eq:allact_sis_dyn}. It follows from Proposition \ref{cor:dt_allact} that when $\delta < \beta \sum_{d \in \DD} dm_d$, there exists a nonzero endemic equilibrium $x^{*}(\DD_p)$ of \eqref{eq:allact_sis_dyn}. Since $\bar{x}^{**}(\DD_p) \leq \frac{1}{L\beta d_p}$, $x^{*}(\DD_p)$ acts as an equilibrium point of \eqref{eq:cl_sis_rational} as well.}  

\rev{On the other hand, if $\bar{x}^{**}(\DD_p) > \frac{1}{L\beta d_p}$, then the dynamics \eqref{eq:cl_sis_rational} no longer coincides with \eqref{eq:allact_sis_dyn}; rather we have $x_{d_p}(k+1) = (1-\delta)x_{d_p}$ for $\bar{x}(k) > \frac{1}{L\beta d_p}$. Then, $x^{*}(\DD_p)$, with $x^{*}_{d_p}(\DD_p) > 0$ is no longer an equilibrium point of \eqref{eq:cl_sis_rational}.} 

\medskip

\noindent{\emph{Case 2: $d_r < d_p$.}} \rev{Recall that $d_r$ is the largest degree for which $\bar{x}^{**}(\DD_{r}) \leq \frac{1}{L\beta d_r}$. Now, for any $\bar{x}(k) \in (\frac{1}{L\beta d_{r+1}},\frac{1}{L\beta d_r}]$, the closed-loop dynamics \eqref{eq:cl_sis_rational} is given by
\begin{equation}\label{eq:thm_int1}
x_d(k+1) = 
\begin{cases}
(1-\delta) x_d(k) + (1-x_d(k)) \sum_{t \in \DD} \beta m_t d x_t(k), \qquad & d \in \DD_r, 
\\ (1-\delta) x_d(k), \qquad & d \in \DD \setminus \DD_r. 
\end{cases}
\end{equation}
Thus, if \eqref{eq:allact_sis_dyn} admits an endemic equilibrium with degrees in $\DD_r$, and $\bar{x}^{**}(\DD_{r}) \in (\frac{1}{L\beta d_{r+1}},\frac{1}{L\beta d_r}]$, then it is easy to see that $\mathbf{x}^*_{cl}$ defined in \eqref{eq:cl_endemic_final} is an equilibrium of the closed-loop dynamics.} 

\rev{We now prove that the above conditions are necessary for any nonzero equilibrium. Let $\mathbf{x}^* \in [0,1]^p$ be a nonzero equilibrium of \eqref{eq:cl_sis_rational} (with $\mathbf{x}^*_d > 0$ for some $d \in \DD$). Let $\bar{x}^* = \sum_{d \in \DD} m_d \mathbf{x}^*_d > 0$. Let $d_k$ be the largest degree such that $\bar{x}^* \leq \frac{1}{L\beta d_k}$, i.e., $\bar{x}^* \in (\frac{1}{L\beta d_{k+1}}, \frac{1}{L\beta d_k}]$ for $d_k < d_p$ and $\bar{x}^* \in [0, \frac{1}{L\beta d_k}]$ for $d_k = d_p$.}

\rev{First, suppose $d_k < d_p$. Note that when $\bar{x}(k) \in (\frac{1}{L\beta d_{k+1}}, \frac{1}{L\beta d_k}]$, the closed loop dynamics is given by 
\begin{equation}\label{eq:thm_int2}
x_d(k+1) = 
\begin{cases}
(1-\delta) x_d(k) + (1-x_d(k)) \sum_{t \in \DD} \beta m_t d x_t(k), \qquad & d \in \DD_k, 
\\ (1-\delta) x_d(k), \qquad & d \in \DD \setminus \DD_k. 
\end{cases}
\end{equation}}

\rev{Thus, for $\mathbf{x}^*$ to be an equilibrium of \eqref{eq:thm_int2}, we must necessarily have $\mathbf{x}^*_d = 0$ for $d > d_k$ and $\{\mathbf{x}^*_d\}_{d \in \DD_k}$ must be an equilibrium of \eqref{eq:allact_sis_dyn} with degrees being $\DD_k$. It follows from Proposition \ref{cor:dt_allact} that \eqref{eq:allact_sis_dyn} has at most one endemic equilibrium given by $x^*(\DD_k)$. As a result, we must have $\bar{x}^* = \bar{x}^*(\DD_k) \in (\frac{1}{L\beta d_{k+1}}, \frac{1}{L\beta d_k}]$. Furthermore, following Proposition \ref{prop:endemic_monotone}, we have $\frac{1}{L\beta d_{k+1}} < \bar{x}^*(\DD_k) \leq \bar{x}^*(\DD_{k+1})$, i.e., $d_k$ is the largest degree for which $\bar{x}^*(\DD_k) \leq \frac{1}{L\beta d_k}$. The necessity part of the proof when $d_k = d_p$ follows from identical arguments as above, and is omitted.}
\end{proof}

%%%%%%%%%%%%%%%%%%%%%%%%%%%%%%%

\end{document}